\documentclass{LMCS}

\def\doi{9(1:04)2013}
\lmcsheading%
{\doi}
{1--29}
{}
{}
{Oct.~27, 2011}
{Feb.~27, 2013}
{}
\usepackage{amssymb}
\usepackage{amsmath, amsthm}
\usepackage{xcolor}
\usepackage{comment}
\usepackage{proof}
\usepackage{xspace}
\usepackage{tikz}
\usepackage{multirow}
\usepackage{graphicx}
\usepackage{pstricks,pst-node}
\usepackage{enumerate,hyperref}

\usetikzlibrary{arrows}

\usetikzlibrary{arrows}

\tikzstyle{node}        =[draw, rounded corners=8, shade,top color=green!10,bottom color=blue!10, minimum size=.5cm, ultra thin]

\tikzstyle{refinement}   =[dotted, above]
\tikzstyle{reduction} =[->]

\newcommand{\subse}{\; \subseteq \;}

\newcommand{\To}[0]{\Rightarrow}
\newcommand{\interp}[1]{[\negthinspace[#1]\negthinspace]}

\newcommand{\ot}[0]{\leftarrow}

\begin{document}

\title[A Rewriting View of Simple Typing]{A Rewriting View of Simple Typing}

\author[A.~Stump]{Aaron~Stump\rsuper a}
\address{{\lsuper{a,c,d}}Computer Science, The University of Iowa}
\email{astump@acm.org, \{garrin-kimmell,roba-elhajomar\}@uiowa.edu}

\author[H.~Zantema]{Hans~Zantema\rsuper b}
\address{{\lsuper b}Department of Computer Science, TU Eindhoven, The Netherlands; and Institute for Computing and Information Sciences, Radboud University, The Netherlands}
\email{h.zantema@tue.nl}

\author[G.~Kimmell]{Garrin~Kimmell\rsuper c}
\address{\vskip-6 pt}

\author[R.~ El Haj Omar]{Ruba~El Haj Omar}
\address{\vskip-6 pt}

\thanks{This work was
    partially supported by the U.S. National Science Foundation, contract
    CCF-0910510, as part of the Trellys project.}

\ACMCCS{[{\bf Software and its Engineering}]: Software notations and tools---
  Formal language definitions---Syntax}

\newtheorem{theorem}[thm]{Theorem}
\newtheorem{lemma}[thm]{Lemma}
\newtheorem{corollary}[thm]{Corollary}
\newtheorem{definition}[thm]{Definition}
\newtheorem{proposition}[thm]{Proposition}
\newtheorem*{theorema}{Theorem}

\begin{abstract}
This paper shows how a recently developed view of typing as small-step
abstract reduction, due to Kuan, MacQueen, and Findler, can be used to
recast the development of simple type theory from a rewriting
perspective.  We show how standard meta-theoretic results can be
proved in a completely new way, using the rewriting view of simple
typing.  These meta-theoretic results include standard type
preservation and progress properties for simply typed lambda calculus,
as well as generalized versions where typing is taken to include both
abstract and concrete reduction.  We show how automated analysis tools
developed in the term-rewriting community can be used to help automate
the proofs for this meta-theory.  Finally, we show how to adapt a
standard proof of normalization of simply typed lambda calculus, for
the rewriting approach to typing.
\end{abstract}

\maketitle

\section{Introduction}
\label{sec:intro}

This paper develops a significant part of the theory of simple types
based on a recently introduced rewriting approach to typing.  The idea
of viewing typing as a small-step abstract reduction relation was
proposed by Kuan, MacQueen, and Findler in 2007, and explored also by
Ellison, \c{S}erb\u{a}nu\c{t}\u{a}, and
Ro\c{s}u~\cite{rosu+10,ellison+08,kuan+07}.  These works sought to use
rewrite systems to specify typing in a finer-grained way than usual
type systems.  Our motivation is more foundational: we seek to prove
standard meta-theoretic properties of type systems directly, based on
the rewriting formulation.  The goal is to develop new methods which
could provide a different perspective on familiar type systems, and
perhaps yield new results for more advanced type systems.

Our focus in this paper is simple type systems, where the central
typing construct is the function type $T \To T'$.  We will view such
types as abstractions of functions, and incrementally rewrite
(typable) functions to such function types, using an abstract
small-step reduction relation.  It will be straightforward to prove
the standard property of type safety, based on type preservation and
progress, using this rewriting formulation.  This viewpoint also
allows us to combine the usual concrete reduction relation and our new
abstract reduction relation together, simply by taking their
set-theoretic union.  We will prove that this combined reduction
relation is confluent for typable terms, defined as terms which
reduce, using abstract steps, to a type.  To prove both type preservation 
and confluence we use observations developed in the context of abstract
reduction systems.  We then develop our final main result, which is a
proof of normalization for the simply typed lambda calculus, based on
the rewriting approach.  This proof has several novel features, which
shed new light on the reducibility semantics of types used in standard
proofs of normalization.

This paper expands in several important ways on a previous paper of
Stump, Kimmell, and El Haj Omar, which was presented at RTA 2011~\cite{stump+11}:
\begin{iteMize}{$\bullet$}
\item 
We use the rewriting method to prove type preservation for full
$\beta$-reduction; the RTA '11 paper showed it only for call-by-value
computation.
\item We prove preservation for a new notion we call generalized
  typing, where concrete and abstract reduction steps can be
  intermixed.  This generalizes the so-called \emph{direct computation
  rules} of the well-known NuPRL system~\cite{allen+06}.
\item We correct an error in the RTA '11 paper, where we claimed that
  type preservation is a corollary of confluence for typable terms.
  In fact, confluence is a straightforward corollary of type preservation.
\item We have shown how a standard proof of normalization for simply
  typable terms is adapted to the rewriting approach to typing.  This
  adaptation reveals an interesting perspective on types as
  abstractions of terms.
\item Due to the amount of new material, we have dropped the treatment
  of several variants of STLC, which are studied in the RTA paper.
\end{iteMize}
As Zantema had a substantial contribution to these extensions, he was
added as an author.  

The remainder of the article is organized as
follows. Section~\ref{sec:rewriting-prelim} provides a brief
introduction to abstract reduction systems as used later in the
paper. Section~\ref{sec:std-presentation} gives a standard
presentation of the simply typed lambda calculus along with the
fundamental meta-theoretic properties. Section~\ref{sec:restlc}
recasts the simply typed lambda calculus static and operational
semantics within the framework of abstract reduction systems.
Section~\ref{sec:ars} gives some abstract reduction theory to be used
in Section~\ref{sec:presstlc} where type preservation and confluence
is proved.  Section~\ref{sec:progtpsafe} then proves progress and type
safety.  Section~\ref{sec:unistc} proves type preservation and
confluence for a system with uniform syntax for types and term.  For
this result, we use automated tools developed in the term-rewriting
community, to verify some of the properties necessary for applying
theorems proved in Section~\ref{sec:ars}.
Section~\ref{sec:genpresstlc} extends these to a generalized notion of
typing, based on the union of the concrete and abstract reduction
relations.  Section~\ref{sec:normstlc} applies a rewriting approach to
prove the normalization of well-typed simply typed lambda calculus
terms. We conclude and identify future directions in
Section~\ref{sec:conclusion}.

\section{Rewriting Preliminaries}
\label{sec:rewriting-prelim}

In this section we collect some basic properties in the setting of
abstract reduction systems.  That is, we consider relations $\to$
being a subset of $X \times X$ for some arbitrary set $X$.

We write $\cdot$ for relation composition, and inductively define
$\to^0 = id$ (the identity) and $\to^n = \to^{n-1} \cdot \to$ for $n > 0$.
As usual, for a relation $\to$ we write $\leftarrow$ for its reverse,
$\to^=$ for its reflexive closure (zero or one times),
$\to^+ = \bigcup_{i=1}^\infty \to^i$ for its transitive closure
(one or more times), and $\to^* = \bigcup_{i=0}^\infty \to^i$ for its
transitive reflexive closure (zero or more times).  We will also
use standard notation $R(A)$ for the image of set $A$ under relation $R$:
\[
R(A) = \{ a'\ |\ \exists a\in A. (a,a')\in R\}
\]
\noindent We can use this notation to denote the set of predecessors
of a set $A$ with respect to $\to$ as $\leftarrow^*(A)$.  We will also
write $\textit{Id}_A$ for $\{(a,a)\ |\ a\in A\}$.

\vspace{3mm}
\noindent A relation $\to$ is said to
\begin{iteMize}{$\bullet$}
\item be {\em confluent}
(Church Rosser, $CR(\to)$) if $\leftarrow^* \cdot \to^* \subse
\to^* \cdot \leftarrow^*$,
\item be {\em locally confluent}
(Weak Church Rosser, $WCR(\to)$) if $\leftarrow \cdot \to \subse
\to^* \cdot \leftarrow^*$,
\item have the {\em diamond property}
($\diamond(\to)$) if $\leftarrow \cdot \to \subse
\to^= \cdot \leftarrow^=$,
\item be {\em deterministic} (det$(\to)$) if $\leftarrow \cdot \to \subse id$.
\item be {\em terminating} if there is no infinite descending chain $a_1\to a_2 \to \cdots$.
\item be {\em convergent} if it is confluent and terminating.
\end{iteMize}

\noindent We will sometimes also call an element $x_1\in X$ confluent
iff for all $x_2,x_3\in X$ with $x_1\to^* x_2$ and $x_1\to^* x_3$,
there exists $x_4\in X$ with $x_2\to^* x_4$ and $x_3\to^* x_4$.  It is
well-known and easy to see that det$(\to) \Rightarrow \diamond(\to)
\Rightarrow CR(\to) \Rightarrow WCR(\to)$.

\ 

\noindent Finally, if $\to_a$ and $\to_b$ are binary relations, below
we will often write $\to_{ba}$ for $\to_a \cup \to_b$.

\section{A Standard Presentation of Simple Typing}
\label{sec:std-presentation}

In this section, we summarize a standard presentation of the simply
typed lambda calculus (STLC), including syntax and semantics, and statements
of the basic meta-theoretic properties of type preservation and
progress.  Sections~\ref{sec:restlc} and following will recapitulate
this development in detail, from the rewriting perspective.  Including
some type and term constants, together with reduction rules for them,
is very standard in the study of programming languages and typed
lambda calculus.  One example is Mitchell's treatment of STLC with
additional rules~\cite[Section 4.4.3]{M96}).  For progress, it is
indeed instructive to include reduction rules for some selected
constants.  Otherwise, there are no stuck terms that should be ruled
out by the type system, since in pure STLC, every closed normal form
is a value, namely a $\lambda$-abstraction.  We treat additional rules
representatively (as opposed to parametrically), using constants $a$
and $f$ below.

\subsection{Syntax and Semantics}
\label{sec:stlcstand}

The syntax for terms, types, and typing contexts is the following,
where $A$, $f$, and $a$ are specific constants, and $x$ ranges over a
countably infinite set of variables:
\[
\begin{array}{lll}
\textit{types}\ T & ::= & A\ |\ T_1 \To T_2 \\
\textit{standard terms}\ t & ::= & f\ |\ a\ |\ x\ |\ t_1\ t_2\ |\ \lambda x:T.t \\
\textit{typing contexts}\ \Gamma & ::= & \cdot\ |\ \Gamma, x:T
\end{array}
\]
\noindent We will write \textit{Types} for the set of all types.  We
assume standard additional conventions and notations, such as
$[t/x]t'$ for the capture-avoiding substitution of $t$ for $x$ in
$t'$, and $E[t]$ for grafting a term into an evaluation context.
Figure~\ref{fig:stlc} defines a standard type system for STLC.  The
judgments derived by the rules in the figure are of the form
$\Gamma\vdash t:T$, which can be viewed as deterministically computing
a type $T$ as output, given a term $t$ and a typing context $\Gamma$
as inputs.  In the topmost leftmost rule of the Figure, we use the
notation $\Gamma(x) = T$ to mean that there is a binding $x:T$ in
$\Gamma$.  We assume there is at most one such binding in $\Gamma$,
renaming bound variables as necessary to ensure this.  A standard
small-step reduction semantics, for unrestricted $\beta$-reduction, is
defined using the rules of Figure~\ref{fig:stlcopsem}.  Following
standard usage, terms of the form $(\lambda x:T. t)\ t'$ or $f\ a$ are
called redexes.  An example of a concrete reduction is (with redexes
underlined):
\[
\underline{(\lambda x:(A\to A).x\ (x\ a))\ f}\ \to\ f\ \underline{(f\ a)}\ \to\ \underline{f\ a}\ \to\ a
\]

\begin{figure}
\[
\begin{array}{lllll}
\infer{\Gamma\vdash x : T}{\Gamma(x) = T}

&\ \ \ &

\infer{\Gamma\vdash f : A\To A}{\ }

&\ \ \ &

\infer{\Gamma\vdash a : A}{\ }

\\ \\

\infer{\Gamma\vdash t_1\ t_2 : T_1}
      {\Gamma\vdash t_1 : T_2 \To T_1 &
       \Gamma\vdash t_2 : T_2}

&\ \ \ &

\infer{\Gamma\vdash \lambda x : T_1.\,t : T_1 \To T_2}
      {\Gamma,x:T_1\vdash t:T_2}

&\ \ \ & \

\end{array}
\]
\caption{Type-computation rules for STLC with selected constants}
\label{fig:stlc}
\end{figure}

\begin{figure}
\[
\begin{array}{ll}
\begin{array}{l}
\infer{E[(\lambda x:T.\,t)\ t']\ \to\ E[[t'/x]t]}{\ } \\ \\
\infer{E[f\ a]\ \to\ E[a]}{\ }
\end{array}

&

\begin{array}{rll}
\textit{values}\ v & ::= & \lambda x:T. t \ |\ a\ |\ f \\
\textit{evaluation contexts}\ E & ::= & *\ |\ (E\ t)\ |\ (t\ E)\ |\ \lambda x:T.\, E\ \\
\ \\
\end{array}
\end{array}
\]
\caption{Small-step reduction semantics for STLC}
\label{fig:stlcopsem}
\end{figure}

\subsection{Basic Meta-theory}
\label{sec:basicmeta}

The main theorem relating the reduction relation $\to$ and typing is
\textbf{type preservation}, which states the following, either for
unrestricted $\beta$-reduction $\to$ or for some restriction of $\to$
(as we will consider below):
\[
(\Gamma\vdash t:T\ \ \wedge\ \ t\, \to\, t')\ \ \To\ \ \Gamma \vdash t':T
\]

\noindent The standard proof method is to proceed by induction on the
structure of the typing derivation, with case analysis on the
reduction derivation (cf. Chapters 8 and 9 of~\cite{pierce02}).  A separate
induction is required to prove a substitution lemma, needed critically
for type preservation for $\beta$-reduction steps:
\[
\Gamma\vdash t : T\ \ \wedge\ \ \Gamma, x:T \vdash t':T'\ \ \To\ \ \Gamma \vdash[t/x]t':T'
\]

\noindent For call-by-value programming languages, one also typically
proves \textbf{progress}, formulated in terms of values:
\[
( \cdot \vdash t:T\ \ \wedge\ \ t\,\not\to)\ \ \To\ \ t \in\textit{values}
\]
\noindent Here, the notation $t\not\to$ means $\forall
t'.\ \neg(t\,\to\, t')$; i.e., $t$ is a normal form.  Normal forms
which are not values are called \emph{stuck} terms.  An example is
$f\ f$.  Combining type preservation and progress allows us to
prove \textbf{type safety}~\cite{wright+94}.  This property states
that the normal forms of closed well-typed terms are values, not stuck
terms, and in our setting can be stated:
\[
(\cdot \vdash t:T\ \ \wedge\ \ t\, \to^*\, t'\,\not\to)\ \ \To\ \ \exists v.\ t' = v
\]
\noindent This is proved by induction on the length of the reduction
sequence from $t$ to $t'$.  As already noted, without constants ($f$
and $a$ here), this result is not so interesting for STLC, since it
follows already by simpler reasoning: reduction cannot introduce new
free variables, so $t'$ must be closed; and it is then easy to prove
that closed normal forms are $\lambda$-abstractions, and hence values
by definition.

\section{Simple Typing as Abstract Reduction}
\label{sec:restlc}

In this section, we see how to view a type-computation (also called
type-synthesis) system for STLC as an abstract operational semantics.
We view function types $T_1\To T_2$ as abstract functions from $T_1$
to $T_2$, and allow these to be applied to arguments.  When $T_1\To
T_2$ is applied to the abstract term $T_1$, an abstract
$\beta$-reduction step is possible, simulating concrete
$\beta$-reduction for any function of type $T_1\To T_2$ applied to an
argument of type $T_1$.  Thus, we will see abstract reduction as truly
an abstraction of the usual reduction, which we thus view, in
contrast, as concrete.

\begin{figure}
\[
\begin{array}{lll}
\textit{types}\ T & ::= & A\ |\ T_1\To T_2\\
\textit{standard terms}\ t & ::= & x\ |\ \lambda x:T.\,t\ |\ t\ t'\ |\ a\ |\ f \\
\textit{mixed terms}\ m & ::= & x\ |\ \lambda x:T.\,m\ |\ m\ m'\ |\ a\ |\ f\ | \\
\ &\ & A\ |\ T\To m \\
\textit{standard values}\ v & ::= & \lambda x:T.t\ |\ a\ |\ f \\
\textit{mixed values}\ u & ::= & \lambda x:T.m\ |\ T\To m\ |\ A\ |\ a\ |\ f
\end{array}
\]
\caption{Syntax for STLC using mixed terms}
\label{fig:synrestlc}
\end{figure}

\begin{figure}
\[
\begin{array}{l}
\begin{array}{lll}
\infer[\textit{c}(\textit{f-}\beta)]{E_c[f\ a]\ \to_c\ E_c[a]}{\ }
&\ &
\infer[\textit{c}(\beta)]{E_c[(\lambda x:T.\,m)\ u]\ \to_c\ E_c[[u/x]m]}{\ }
\\ \\
\infer[\textit{b}(\textit{f-}\beta)]{E_a[f\ a]\ \to_b\ E_a[a]}{\ }
&\ &
\infer[\textit{b}(\beta)]{E_a[(\lambda x:T.\,m)\ m']\ \to_b\
E_a[[m'/x]m]}{\ }
\\ \\
\infer[\textit{a}(\beta)]{E_a[(T \To m)\ T]\ \to_a\ E_a[m]}{\ }
&\ &\
\infer[\textit{a}(\lambda)]{E_a[\lambda x:T.\, m]\ \to_a\ E_a[T\To [T/x]m]}{\ }
\\ \\
\infer[\textit{a}(f)]{E_a[f]\ \to_a\ E_a[A\To A]}{\ }
&\ &\
\infer[\textit{a}(a)]{E_a[a]\ \to_a\ E_a[A]}{\ }

\end{array}
\\ \\ \\

\begin{array}{lll}
\textit{call-by-value evaluation contexts}\ E_c & ::= & *\ |\ (E_c\ m)\ |\ (u\ E_c) \\
\textit{unrestricted evaluation contexts}\ E_a & ::= & *\ |\ (E_a\ m)\ |\ (m\ E_a)\ |\ \lambda x:T.\, E_a\ |\ T\To E_a \\
\end{array}
\end{array}
\]
\caption{Concrete call-by-value reduction ($\to_c$), concrete full $\beta$-reduction ($\to_b$), and abstract reduction ($\to_a$) for STLC}
\label{fig:restlc}
\end{figure}

To view typing as an abstract form of reduction, we use mixed terms,
defined in Figure~\ref{fig:synrestlc}.  Types like $T_1\To T_2$ will
serve as abstractions of $\lambda$-abstractions.  For our development
below, we are going to consider both unrestricted $\beta$-reduction,
and also call-by-value $\beta$-reduction, a common restriction
implemented in practical functional programming languages like
\textsc{OCaml}.  Figure~\ref{fig:restlc} gives rules for concrete
call-by-value reduction ($\to_c$), concrete full $\beta$-reduction
($\to_b$), and abstract reduction ($\to_a$).  As above, we will refer
to any term of the form displayed in context on the left hand side of
the conclusion of a rule as a redex. We denote the union of these
reduction relations as $\to_{ca}$.  The definition of call-by-value
evaluation contexts $E_c$ enforces left-to-right evaluation order in a
standard way, while unrestricted evaluation contexts $E_a$ make
abstract reduction and full $\beta$-reduction non-deterministic:
reduction is allowed anywhere inside a term.  This is different from
the approach followed by Kuan et al., where abstract and concrete
reduction are both deterministic.  Here is an example of reduction using
the abstract operational semantics:
\[
\begin{array}{l}
\lambda x:(A\To A).\,\lambda y:A.\, (x\ (x\ y)) \ \to_a \\
\lambda x:(A\To A).\, A\ \To\ (x\ (x\ A)) \ \to_a \\
(A\To A)\ \To\ A\ \To ((A\To A)\ ((A\To A)\ A)) \ \to_a \\
(A\To A)\ \To\ A\ \To ((A\To A)\ A) \ \to_a \\
(A\To A)\ \To\ A\ \To A\ \not\to_a
\end{array}
\]
\noindent The final result is a type $T$, which does not reduce (as
noted below).  Indeed, using the standard typing rules of
Section~\ref{sec:stlcstand}, we can prove that the starting term of
this reduction has that type $T$, in the empty typing context.
Abstract reduction to a type plays the role of typing above.

\begin{lem}
\label{lem:tpnorm}
For all types $T$, we have $T\not\to_a$.
\end{lem}
\begin{proof} This follows by induction on $T$ and inspection of the rules for $\to_a$.
\end{proof}

If we look back at our standard typing rules (Figure~\ref{fig:stlc}),
we can now see them as essentially big-step abstract operational
rules.  Recall that big-step call-by-value operational semantics for
STLC includes this rule (as well as several others which we elide):
\[
\infer{t_1\ t_2\ \Downarrow\ t'}{t_1\ \Downarrow\ \lambda x:T.t_1' & t_2\ \Downarrow\ t_2' & [t_2'/x]t_1'\ \Downarrow\ t'}
\]
\noindent In our setting, big-step call-by-value semantics would be
seen as a concrete big-step reduction, which we might denote
$\Downarrow_c$.  The abstract version of this rule, where we abstract
$\lambda$-abstractions by arrow-types, is
\[
\infer{t_1\ t_2\ \Downarrow_a\ T'}{t_1\ \Downarrow_a\ T\To T' & t_2\ \Downarrow_a\ T}
\]
\noindent If we drop the typing context from the standard typing rule
for applications (in Figure~\ref{fig:stlc}), we obtain essentially
the same rule.

The standard approach to proving type preservation relates a
small-step concrete operational semantics with a big-step abstract
operational semantics (i.e., the standard typing relation).  We find
it both more elegant, and arguably more informative to relate abstract
and concrete small-step relations, as we will do in Section~\ref{sec:presstlc} below.

\subsection{Rewriting Properties of Abstract Reduction}

In this subsection, we study the properties of abstract reduction from
the perspective of the theory of abstract reduction systems (ARSs).
From this point of view, abstract reduction is very well behaved: it
is a convergent ARS, as the following two theorems show.

\begin{thm}[Termination of Abstract Reduction]
\label{thm:termabstr}
The relation $\to_a$ is terminating.
\end{thm}
\begin{proof} We recursively define a natural-number measure $\mu(m)$ which can be confirmed to reduce from $m$ to $m'$ 
whenever $m\to_a m'$:
\begin{eqnarray*}
\mu(x) & = & 1 \\
\mu(\lambda x:T.m) & = & 1+\mu(m) \\
\mu(m\ m') & = & 1+\mu(m)+\mu(m') \\
\mu(a) & = & 1 \\
\mu(f) & = & 1 \\
\mu(A) & = & 0 \\
\mu(T \To m) & = & \mu(m)
\end{eqnarray*}
\end{proof}

\begin{thm}
\label{lem:confl}
The relation $\to_a$ is confluent.
\end{thm}
\begin{proof}
In fact, we will prove $\to_a$ has the diamond property (and hence is
confluent).  Suppose $m \to_a m_1$ and $m\to_a m_2$.  No critical
overlap is possible between these steps, because none of the redexes
in the $a$-rules of Figure~\ref{fig:restlc} (such as $(T\To m)\ T$ in
the $a(\beta)$ rule) can critically overlap another such redex.  If
the positions of the redexes in the terms are parallel, then (as
usual) we can join $m_1$ and $m_2$ by applying to each the reduction
required to obtain the other.  Finally, we must consider the case of
non-critical overlap (where the position of one redex in $m$ is a
prefix of the other position).  We can also join $m_1$ and $m_2$ in
this case by applying the reduction to $m_i$ which was used in $m
\to_a m_{3-i}$, because abstract reduction cannot duplicate or delete
an $a$-redex.  The only duplication of any subterm in the abstract
reduction rules of Figure~\ref{fig:restlc} is of the type $T$ in
$\textit{a}(\lambda)$.  The only deletion possible is of the type $T$
in $\textit{a}(\beta)$.  Since types cannot contain redexes, there is
no duplication or deletion of redexes.  This means that if the
position of the first redex is a prefix of the second (say), then
there is exactly one descendant (see Section 4.2 of~\cite{terese}) of
the second redex in $m_1$, and this can be reduced in one step to join
$m_1$ with the reduct of $m_2$ obtained by reducing the first redex.
So every \textit{aa}-peak can be completed with one joining step on
each side of the diagram.  This gives the diamond property (and thus
confluence for $\to_a$).
\end{proof}

\subsection{Relation with Standard Typing}

In this subsection, we prove the following theorem, which relates our
notion of typing with the standard one.  The proof begins after the
statement of some simple auxiliary lemmas, whose proofs are routine
and omitted.  The proof of the right-to-left direction of the
implication will take advantage of the fact that abstract reduction is
convergent, as proved in the previous subsection.

\begin{thm}
\label{thm:relatetyp}
For standard terms $t$, a typing judgment $x_1:T_1,\cdots,x_n:T_n \vdash t:T$ holds iff $[T_1/x_1,\cdots,T_n/x_n]t \to_a^* T$.
\end{thm}

\begin{lem}
\label{lem:lemma1}
If $t_1 \to_a^k T$ , then
$t_1\ t_2 \to_a^k T\ t_2$.
\end{lem}
\begin{lem}
\label{lem:lemma2}
If $t_2 \to_a^k T$ , then
$t_1\ t_2 \to_a^k t_1\ T$.
\end{lem}
\begin{lem}
\label{lem:lemma3}
If $t \to_a^k T'$ , then $T \To t \to_a^k T \To T'$.
\end{lem}

\proof[Proof of Theorem~\ref{thm:relatetyp}, left-to-right]
Suppose $x_1:T_1,\cdots,x_n:T_n \vdash t:T$.
We will now prove $[T_1/x_1,\cdots,T_n/x_n]t \to_a^* T$ by
induction on the structure of the typing derivation of $t$.
To simplify the writing of the proof, we will use the
following notation:
\[
\begin{array}{lll}
\Gamma & = & x_1:T_1,\cdots,x_n:T_n \\
\Gamma_{sub} & = & [T_1/x_1,\cdots,T_n/x_n]
\end{array}
\]

\

\noindent \textbf{Base Case:}
\[
\infer{\Gamma\vdash x : T}{\Gamma(x) = T}
\]
\noindent There must be some $i\in\{1,\ldots,n\}$ such that $x = x_i$ and $T = T_i$.
So $\Gamma_{sub}\ x = T_i \to_a^* T_i$ as required.

\

\noindent\textbf{Base Case:}
\[
\infer{\Gamma\vdash f : A\To A}{\ }
\]
\noindent We indeed have $f\to_a (A\To A)$, as required.  The case for
$a:A$ is similar.

\

\noindent \textbf{Case:}
\[
\infer{\Gamma\vdash t_1\ t_2 : T_1}
      {\Gamma\vdash t_1 : T_2 \To T_1 &
       \Gamma\vdash t_2 : T_2}
\]
\noindent By the induction hypotheses for the derivations given for
the two premises of this rule, we have:
\[
\begin{array}{l}
\Gamma_{sub}\ t_1 \to_a^* T_2 \To T_1 \\
\Gamma_{sub}\ t_2 \to_a^* T_2
\end{array}
\]
\noindent Our goal now is to construct the reduction sequence:
\[
\Gamma_{sub}\ (t_1\ t_2) \to_a^* (T_2 \To T_1)\Gamma_{sub}\ t_2 \to_a^*
(T_2 \To T_1)T_2 \to_a T_1
\]
\noindent To construct this sequence, it is sufficient to apply transitivity of
$\to_a^*$ and Lemmas~\ref{lem:lemma1} and~\ref{lem:lemma2}.

\ 

\noindent \textbf{Case:}
\[
\infer{\Gamma\vdash \lambda x : T.\,t : T \To T'}
      {\Gamma,x:T\vdash t:T'}
\]
\noindent By the induction hypothesis on the premise of this rule,
we have:
\[
\Gamma_{sub}\ [T/x]\ t \to_a^* T'
\]
\noindent Now we need to show that
\[
\Gamma_{sub}\ (\lambda x : T.\,t) \to_a^* (T \To T')
\]
\noindent By applying one $a(\lambda)$ step and Lemma~\ref{lem:lemma3}
we get:
\[
\Gamma_{sub}\ (\lambda x : T.\,t) \to_a
(T \To \Gamma_{sub}\ [T/x]\ t) \to_a^* (T \To T')
\]
\noindent This requires the fact that $\Gamma_{sub}\ [T/x] =
          [T/x]\Gamma_{sub}$, which holds because
          $x\not\in\textit{dom}(\Gamma_{sub})$ since we may rename $x$
          to avoid this, and because $T$ contains no term variables
          and hence is unaffected by applying $\Gamma_{sub}$.\qed

\proof[Proof of Theorem~\ref{thm:relatetyp}, right-to-left]
Since abstract reduction is convergent (Theorems~\ref{thm:termabstr}
and~\ref{lem:confl}), we may assume that redexes in the reduction
sequence to $T$ are always reduced in leftmost order.  Note that
convergence is sufficient to justify this assumption, as $T$ is a
normal form, and hence any strategy is guaranteed to reduce the
starting term to $T$ in a finite number of steps.  This assumption
will simplify some reasoning below.  We assume
$[T_1/x_1,\cdots,T_n/x_n]t \to_a^* T$ and prove
$x_1:T_1,\cdots,x_n:T_n \vdash t:T$ by induction on the number $n$ of
leftmost $\to_a$ steps in the reduction to $T$.

\

\noindent \textbf{Base Case:} there are no $\to_a$ steps.
This means that our term $t$ cannot be reduced
\[
\Gamma_{sub}\ t = T
\]
\noindent In this case, $t$ must be a variable (or else substitution could not
result in a type $T$). So, $t=x$ for some variable $x$, where
$\Gamma(x)=T$. Then we get:
\[
\infer{\Gamma\vdash x : T}{\Gamma(x) = T}
\]

\noindent \textbf{Step Case:}
there is at least one $\to_a$ step.  We proceed by case splitting on the form of $t$.

\

\noindent\textbf{Case:}
\[
\Gamma_{sub}\ x
\]
\noindent This case cannot occur, since either
$x\not\in\textit{dom}(\Gamma_{sub})$, in which case we cannot have
$x\to_a^* T$ for any type $T$; or else
$x\in\textit{dom}(\Gamma_{sub})$, and then $\Gamma_{sub}\ x = T$.  We
cannot have a $\to_a$ step in that case, because types are normal
forms for abstract reduction (Lemma~\ref{lem:tpnorm}).

\noindent\textbf{Case:}
\[
\Gamma_{sub}\ f
\]
\noindent The only possible step is $f \to_a A\To A$, and we indeed have
$\Gamma \vdash f : A\To A$.  The case for $\Gamma_{sub}\ a$ is similar.

\noindent\textbf{Case:}
\[
\Gamma_{sub}\ (t_1\ t_2)
\]
\noindent In this case, the reduction sequence must be of the following form,
for some mixed term $t'$ and type $T_2$, and some natural numbers $k_1$ and $k_2$:
\[
\Gamma_{sub}\ (t_1\ t_2) \to_a^{k_1} ((T_2 \To T)\ t_2) \to_a^{k_2} (T_2\To T)\ T_2 \to_a T
\]
\noindent where
\[
\begin{array}{llll}
1. & \Gamma_{sub}\ t_1 &\to_a^{k_1}& T_2 \To T\\
2. & \Gamma_{sub}\ t_2 &\to_a^{k_2}& T_2 
\end{array}
\]
\noindent We are justified in assuming this, because there must be
some first position in the reduction sequence from $t_1\ t_2$ to $T$
where a descendant of $t_1\ t_2$ is reduced.  That descendant here is
$(T_2\To T)\ T_2$.  In the reduction sequence prior to that point, we
are assuming (as noted at the start of the proof) that steps occur in
leftmost order, so the $t_1$ steps come first, and then the $t_2$
ones.  Now we can apply the induction hypothesis to (1) and (2), which
each have shorter length than the original reduction sequence.  This
gives us the premises of the following inference, which suffices to
complete this case:
\[
\infer{\Gamma \vdash t_1\ t_2 : T}{\Gamma \vdash t_1 : T_2 \To T & \Gamma \vdash t_2 : T_2}
\]

\noindent \textbf{Case:}
\[
\Gamma_{sub}\ (\lambda x : T'.\,t')
\]
\noindent In this case, we may assume the reduction sequence is of the
following form, for some $T''$:
\[
\Gamma_{sub}\ (\lambda x : T'.\,t') \to_a
(T' \To [T'/x] \Gamma_{sub}\ t') \to_a^* (T' \To T'')
\]
\noindent where
\[
[T'/x] \Gamma_{sub}\ t' \to_a^* T''
\]
\noindent This is because $\lambda x:T'.\,t'$ is itself an abstract
redex, and since we are assuming our reduction is in leftmost,
it must be reduced immediately.  Now we can apply the induction
hypothesis on $[T'/x] \Gamma_{sub}\ t' \to_a^* T''$ and get the
premise of the following inference, which suffices to complete this
case:

\[
\infer{\Gamma\vdash \lambda x : T'.\,t' : T' \To T''}
      {\Gamma,x:T'\vdash t':T''}\eqno{\qEd}
\]

\section{Generic Theorems for Preservation and Combined Confluence}
\label{sec:ars}

In this section, we collect some abstract properties for $\to_a$ and
$\to_b$, from which type preservation and confluence of $\to_{ab}$ can
be concluded. In subsequent sections we will instantiate these
theorems with abstract and concrete reduction relations.

For the first theorem, recall that in our setting $\to_a$ computes the
type of a term, or else could reach a stuck term like $(A\To A)\ (A\To
A)$ which does not correspond to a type.  We want to speak about
reductions that lead to types, so we need to phrase the following
theorem in terms of some set $S$, which we will instantiate later with
a set of types.  In condition (3) of the theorem, we interpose
$\textit{Id}_{\leftarrow_a^*(S)}$ to restrict peaks to those objects
which $a$-reduce to an object in $S$.

\begin{thm}
\label{thmtp}
Assume
\begin{enumerate}[\em(1)]
\item $\to_a(S) = \emptyset$ (that is, $S$ is a set of objects in normal form with respect to $\to_a$).
\item $\to_a$ is confluent.
\item $\leftarrow_a \cdot \textit{Id}_{\leftarrow_a^*(S)} \cdot \to_b
  \subse (\to_b \cup \to_a^*) \cdot \leftarrow_a^*$; that is, for
  every $m$ such that there exists $T\in S$ with
  $m\rightarrow_a^* T$, and every $m'$ and $m''$ with $m\rightarrow_a
  m'$ and $m\rightarrow_b m''$, there exists a $m'''$ such that
  $m''\rightarrow_a^* m'''$ and either $m'\rightarrow_b m'''$ or
  $m'\rightarrow_a^* m'''$.
\item every normal form with respect to $\to_a$
is also a normal form with respect to $\to_b$.
\end{enumerate}
Then if $T\in S$ and $T \leftarrow_a^* m \to_b m'$, we have $m' \to_a^* T$.
\end{thm}

\begin{proof}
Let $m \to_a^* T$ and $m\to_b m'$, we have to prove that $m'\to_a^*
T$.  We do this by induction on the number $n$ of steps in $m \to_a^n
T$. In case $n=0$ we have $m = T$.  By (1), $T$ is a normal form with
respect to $\to_a$, which is a normal form with respect to $\to_b$ due
to (4). So $m\to_b m'$ is not possible, and the claim holds
trivially.

For the induction step assume $m \to_a m_1$ for which $m_1 \to_a^{n-1}
T$. Applying (3) now yields $m_3$ such that $m' \to_a^* m_3$ and
either $m_1 \to_b m_3$ or $m_1 \to_a^* m_3$. In case $m_1 \to_b m_3$
we apply the induction hypothesis on $m_1 \to_a^{n-1} T$ and conclude
$m' \to_a^* m_3 \to_a^* T$.  In case $m_1 \to_a^* m_3$ we apply
confluence of $\to_a$ (2) by which $T$ and $m_3$ have a common
$\to_a$-reduct. As $T$ is a normal form with respect to $\to_a$ by
(1), we conclude $m' \to_a^* m_3 \to_a^* T$, concluding the proof.
\end{proof}

\begin{lem}
\label{lem:extendac}
Suppose $\to_a$ and $\to_b$ are binary relations such that
\begin{enumerate}[\em(1)]
\item $\to_a$ is confluent, and
\item $\leftarrow_a \cdot \to_b \subse (\to_b \cup \to_a^*) \cdot \leftarrow_a^*$.
\end{enumerate}
Then we also have
\[ \leftarrow_a^* \cdot \to_b \subse (\to_b \cup \to_a^*) \cdot \leftarrow_a^*\]
\end{lem}
\begin{proof}
Assume $t \to_a^n u$ and $t \to_b v$; we have to find $w$ such that $u
\to_b \cup \to_a^* w$ and $v \to_a^* w$. We do this by induction on
$n$. For $n=0$ we choose $w = v$. For $n > 0$ write $t \to_a t'
\to_a^{n-1} u$. By (2) an element $v'$ exists such that $v \to_a^* v'$
and either $t' \to_a^* v'$ or $t' \to_b$. If $t' \to_a^* v'$ we apply
(1) yielding $w$ satisfying $u \to_a^* w$ and $v' \to_a^* w$ and we
are done. If $t' \to_b$ then we apply the induction hypothesis
yielding $u (\to_b \cup \to_a^*) w$ and $v' \to_a^* w$.
\end{proof}

\begin{thm}
\label{thmcr}
Let $\to_a$ and $\to_b$ be binary relations (recall from
Section~\ref{sec:rewriting-prelim} that we write $\to_{ba}$ for $\to_a
\cup \to_b$). Assume
\begin{enumerate}[\em(1)]
\item $\to_a$ is terminating,
\item $\to_a$ is confluent,
\item $\leftarrow_a \cdot \to_b \subse (\to_b \cup \to_a^*) \cdot \leftarrow_a^*$, and
\item every normal form with respect to $\to_a$
is also a normal form with respect to $\to_b$.
\end{enumerate}
Then $\to_{ba}$  is confluent.
\end{thm}

\begin{proof} 
By Lemma~\ref{lem:extendac}, we have:
\[ (3') \;\; \leftarrow_a^* \cdot \to_b \subse (\to_b \cup \to_a^*) \cdot \leftarrow_a^*.\]
\noindent Now let $t \to_{ba}^* u$ and $t \to_{ba}^* v$; for proving
the theorem we have to prove that $w$ exists satisfying $u \to_{ba}^*
w$ and $v \to_{ba}^* w$. Choose $w$ to be a $\to_a$-normal form of
$t$, which exists due to (1).  Assume $t \to_{ba}^n u$; we will prove
that $u \to_a^* w$ by induction on $n$. For $n = 0$ this follows from
$t \to_a^* w$. For $n > 0$ let $t \to_{ba}^{n-1} u' \to_{ba} u$. From
the induction hypothesis we conclude $u' \to_a^* w$.  Combining (2)
and $(3')$ yields
\[\leftarrow_a^* \cdot \to_{ba} \subse (\to_b \cup \to_a^*) \cdot \leftarrow_a^*.\]
So since $w \leftarrow_a^* u' \to_{ba} u$ we conclude that $w'$ exists
satisfying $w  \to_b w'$ or $w \to_a^* w'$, and $u \to_a^* w'$. Since $w$ is not only a
$\to_a$-normal form, but also a $\to_b$-normal form according to (4), 
we conclude $w' = w$. Hence $u \to_a^* w' = w$,
concluding the  proof of $u \to_a^* w$. Applying the same argument on
$t \to_{ba}^* v$ we conclude $v \to_a^* w$, concluding the proof of the
theorem. 
\end{proof}

\vspace{3mm}

One may wonder whether the requirement of
termination is essential for Theorem \ref{thmcr}. It is: on the set $\{1,2,3\}$ the
relations $\to_a = \{(1,1)\}$ and $\to_b = \{(1,2),(1,3)\}$ satisfy
all requirements of Theorem \ref{thmcr}, while $\to_{ba}$ is not confluent. 

One may wonder whether in Theorem \ref{thmcr} the condition (4) on normal
forms is essential. It is, even if not only $\to_a$ is terminating and
confluent but also $\to_b$, as is shown by the following example of
relations on 10 elements, in which $\to_a$ steps are denoted by 
dashed arrows and $\to_b$ steps are denoted by solid arrows.

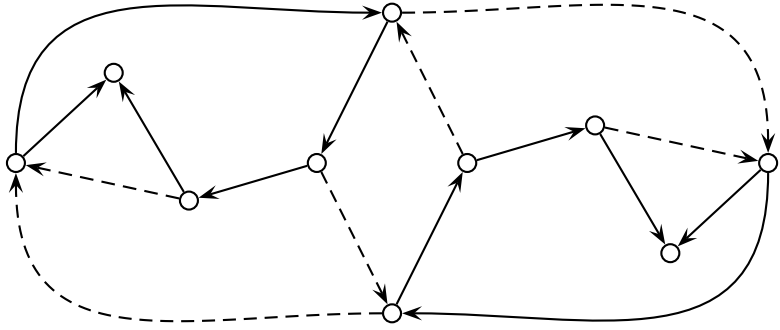
\vspace{4mm}

\begin{center}
\begin{pspicture}(0,0)(10,4)%
\psset{arrows=->,arrowsize=3pt 3}
\rput(1.3,3.2){\circlenode{2}{}}
\rput(0,2){\circlenode{1}{}}
\rput(5,4){\circlenode{7}{}}
\rput(6,2){\circlenode{6}{}}
\rput(7.7,2.5){\circlenode{9}{}}
\rput(2.3,1.5){\circlenode{10}{}}
\rput(4,2){\circlenode{5}{}}
\rput(5,0){\circlenode{8}{}}
\rput(10,2){\circlenode{4}{}}
\rput(8.7,0.8){\circlenode{3}{}}
\ncline{1}{2}
\nccurve[ncurv=1,angleB=180,angleA=90]{1}{7}
\ncline{10}{2}
\ncline{5}{10}
\ncline{7}{5}
\ncline{4}{3}
\nccurve[ncurv=1,angleB=0,angleA=270]{4}{8}
\ncline{8}{6}
\ncline{9}{3}
\ncline{6}{9}

\psset{linestyle=dashed}
\ncline{5}{8}
\ncline{6}{7}
\nccurve[ncurv=1,angleB=90,angleA=0]{7}{4}
\nccurve[ncurv=1,angleB=270,angleA=180]{8}{1}
\ncline{9}{4}
\ncline{10}{1}
\end{pspicture}
\end{center}

\vspace{4mm}

In this example there are two convertible normal forms, so the union is
not confluent, and both $\to_a$ and $\to_b$ are both confluent and terminating;
$\to_a$ is even deterministic. Also condition $(3)$ of Theorem
\ref{thmcr} is easily checked, even stronger:
$\leftarrow_a \cdot \to_b \subse \to_{ba} \cdot \leftarrow_a^=$.
This example was found using a SAT solver. A direct encoding of the example to be
looked for run out of resources. However, by adding a symmetry requirement, 
was observed on the first example, the SAT solver yielded a satisfying assignment
that could be interpreted as a valid example. The example given
above was obtained from this after removing some redundant arrows.
Independently, Bertram Felgenhauer found an example that could be
simplified to exactly the same example as given here. This remarkable
example was the starting point of developing the tool CARPA by which such
examples can be found fully automatically.

\section{Type Preservation and Combined Confluence for STLC}
\label{sec:presstlc}

We now prove type preservation for full $\beta$-reduction (the $\to_b$
relation of Section~\ref{sec:restlc}), based on the rewriting
formulation.  This is in contrast to the results of Kuan et al., who
obtain type preservation for the rewriting approach as a corollary of
type preservation based on a standard big-step notion of typing (and
the relation of that notion of typing with the small-step notion).  

\begin{defi}[Typability]
A mixed term $m$ is called \textbf{typable} 
if $m\to_a^* T$ for some type $T$.  
\end{defi}

\noindent If we translate our standard statement of type preservation (at the
beginning of Section~\ref{sec:basicmeta}) so that it uses abstract
reduction instead of the usual typing relation, we have the following.

\begin{thm}[Type Preservation]
\label{thm:presstlc}
Let $m,m'$ be mixed terms and $T$ be a type.
If $m \to_a^* T$ and $m\to_b m'$, then $m'\to_a^* T$.
\end{thm}

\noindent The proof of this theorem is given by applying Theorem
\ref{thmtp}: we need to check its conditions (1), (2), (3) and (4). We
instantiate the set $S$ in condition (1) with the set of types $T$,
which are normal forms by Lemma~\ref{lem:tpnorm}.  Condition (2)
follows from Lemma \ref{lem:confl}. Condition (4) is immediate from
the definitions of $\to_a$ and $\to_b$: if $\to_b$ applies on a
term $t$, then $t$ either contains $f a$ via rule
$\textit{b}(\textit{f-}\beta)$ by which $\to_a$ applies via $a(f)$, or
$t$ contains $\lambda x:T.\, m]$ via rule $\textit{b}\beta)$ by which
  $\to_a$ applies via $a(\lambda)$. So it remains to check condition
  (3), which follows from the following lemma.

\begin{lem}
\label{lem:beta}
Let $m_0$ be a typable mixed term and let $m_1,m_2$ be mixed terms
such that $m_0 \to_a m_1$ and $m_0 \to_b m_2$. Then a mixed term $
m_3$ exists such that $m_2 \to_a^* m_3$ and 
either $m_1 \to_b m_3$ or $m_1 \to_a^* m_3$. Furthermore, if the step
from $m_0$ to $m_2$ is a call-by-value step, so is the step from $m_1$
to $m_3$.
\end{lem}

\begin{proof}
We distinguish the ways the redexes in $m_0$ are related.

If the redexes of $m_0 \to_a m_1$ and $m_0 \to_b m_2$ are parallel,
then $m_3$ can be chosen such that $m_1 \to_b m_3$ and $m_2 \to_a m_3$
(preserving whether or not the $b$-step is call-by-value).

If the redex of $m_0 \to_a m_1$ is above the redex of $m_0 \to_b
m_2$,
then the $\to_a$ step is either of the type $a(\beta)$ or
$a(\lambda)$, in which the $\to_b$ acts on the mixed term $m$ as it
occurs in the rule $a(\beta)$ or $a(\lambda)$. As this $m$ is not
duplicated, we get $m_3$ such that $m_1 \to_b m_3$ and
$m_1 \to_a m_3$ (and the step $m_0\to_b m_2$ is not call-by-value).

If the redex of $m_0 \to_a m_1$ is below the redex of $m_0 \to_b
m_2$, then some further case analysis is required.

If there is no overlap, then $m_3$ can be chosen such that 
$m_1 \to_b m_3$ (preserving being call-by-value) and $m_2 \to_a^* m_3$. 

If there is overlap and $m_0 \to_a m_1$ is an application of
$a(f)$ or $a(a)$, then $m_0 = E_a[f\ a]$ and $m_2 = E_a[a]$, and
$m_3$ can be chosen to be $E_a[A]$, satisfying
$m_1 \to_a^2 m_3$ and $m_2 \to_a m_3$.

The remaining case is illustrated by the following picture:

\begin{center}
\begin{tikzpicture}
\draw (0,1.5) node(t) {$m_0 = E_a[(\lambda x:T. m)\ m']$} ;
\draw (-2.5,0) node(sa) {$m_1 = E_a[(T\To [T/x]m)\ m']$};
\draw (2.5,0) node(sb) {$m_2 = E_a[[m'/x]m]$};

\draw[-latex'] (t) -- node[left=1pt,very near end]{$a$} (sa);
\draw[-latex'] (t) -- node[left=1pt,very near end]{$b$}
                      (sb);

\draw (-2.5,-1.5) node (saa) {$E_a[(T\To [T/x]m)\ T]$};
\draw (0,-3) node (b) {$E_a[[T/x]m]$};

\draw[-latex'] (sb) -- node[right=2pt] {\textit{since}\ $m'\to_a^* T$}
                       node[left=1pt,very near end]{$a$}
                       node[right=1pt,very near end]{*}
                       (b);
\draw[-latex'] (sa) -- node[left=2pt] {\textit{since}\ $m'\to_a^* T$}
                       node[left=1pt,very near end]{$a$}
                       node[right=1pt,very near end]{*}
                       (saa);
\draw[-latex'] (saa) -- node[left=1pt,very near end]{$a$} (b);

\end{tikzpicture}
\end{center}

The picture already shows that by choosing $m_3 = E_a[[T/x]m]$ we
obtain $m_1 \to_a^* m_3$ and  $m_2 \to_a^* m_3$ if we can prove 
$m'\to_a^* T$. For doing so we use the assumption that $m_0$ is
typable: there exists a type $T'$ such that 
$m_0 = E_a[(\lambda x:T. m)\ m'] \to_a^* T'$. Since $T'$ is a type 
it does not contain a $\lambda$ symbol, so somewhere in this
reduction the $\lambda$ in $\lambda x:T. m$ should be removed. By
inspecting the rules we see that this can only be done by the rule
$a(\lambda)$ by which $\lambda x:T. -$ is replaced by $T \To -$.
Next the (invisible) application symbol in $(\lambda x:T. m)\ m'$
should be removed. This can only be done by the rule $a(\beta)$.
This rule is only applicable if first $m'$ is rewritten by $\to_a$
steps to $T$, indeed proving $m'\to_a^* T$.
\end{proof}

\begin{thm}
\label{prop:confstlch}
The relation $(\textit{Id}_{\leftarrow_a^*(\textit{Types})}
\cdot \to_a)\cup(\textit{Id}_{\leftarrow_a^*(\textit{Types})}
\cdot \to_b)$ is confluent.
\end{thm}
\begin{proof}
We will apply Theorem \ref{thmcr}.  For this, we need to check
properties (1) to (4) for the particular relations
$\textit{Id}_{\leftarrow_a^*(\textit{Types})} \cdot \to_a$ and
$\textit{Id}_{\leftarrow_a^*(\textit{Types})} \cdot \to_b$. Property
(2) follows from Theorem \ref{lem:confl} and the fact that
$\textit{Id}_{\leftarrow_a^*(\textit{Types})}$ is the identity
relation. All peaks must be of the form $m_1 \leftarrow_a m
\leftarrow_{id} m \to_{id} m \to_a m_2$, due to the composition with
$\textit{Id}_{\leftarrow_a^*(\textit{Types})}$. By Theorem
\ref{lem:confl}, if $m_1 \leftarrow_a m \to_a m_2$, then there exists
$m_3$ such that $m_1 \to_a^* m_3 \leftarrow_a m_2$. Thus, any
$\textit{Id}_{\leftarrow_a^*(\textit{Types})} \cdot \to_a$ peak $m_1
\leftarrow_a m \leftarrow_{id} m \to_{id} m \to_a m_2$ can be
completed with $m_1 \to_{id} m_1 \to_a m_3 \leftarrow_a m_2
\leftarrow_{id} m_2$. Likewise, By Theorem \ref{thm:termabstr} $\to_a$
is terminating, so $\textit{Id}_{\leftarrow_a^*(\textit{Types})} \cdot
\to_a \subse \to_a$ is also terminating, proving property (1). Property (3) follows from
Lemma \ref{lem:beta}. So it remains to prove Property (4). This is
immediate from the definitions of $\to_a$ and $\to_b$: if $\to_b$
applies on a term $t$, then $t$ either contains $f a$ via rule
$\textit{b}(\textit{f-}\beta)$ by which $\to_a$ applies via $a(f)$, or
$t$ contains $\lambda x:T.\, m$ via rule $\textit{b}(\beta)$ by which
$\to_a$ applies via $a(\lambda)$. 
\end{proof}

\begin{cor}[Confluence of Combined Reduction]
\label{prop:confstlc}
Every typable mixed term is confluent with respect to the reduction
relation $\to_{ba}$.
\end{cor}
\begin{proof} Confluence of the set of typable mixed terms is equivalent
to confluence of  the relation $\textit{Id}_{\leftarrow_a^*(\textit{Types})} \cdot \to_{ba}$,
which is easily seen to be equal to
\[
(\textit{Id}_{\leftarrow_a^*(\textit{Types})}
\cdot \to_a)\cup(\textit{Id}_{\leftarrow_a^*(\textit{Types})}
\cdot \to_b)
\]
\noindent By Theorem~\ref{prop:confstlch}, the latter relation is confluent.
\end{proof}

A form of typability is essential, since the relation $\to_{ba}$ is not
confluent in general, as Kuan et al. note also in their setting.  For
instance, the non-typable term $(\lambda x:A.x)(\lambda x:A.x)$ has
two distinct normal forms
\[ (A \To A)(A \To A) \leftarrow_a^+ (\lambda x:A.x)(\lambda x:A.x)
\to_b \lambda x:A.x \to_a (A \To A).\]

\section{Progress and Type Safety for STLC}
\label{sec:progtpsafe}

In this section, we complete the basic meta-theory for STLC by proving
progress and type safety theorems for call-by-value reduction (the
$\to_c$ relation of Section~\ref{sec:restlc}).  Lemmas~\ref{lem:cac}
and~\ref{lem:stuck} are stated in a somewhat more general way, so that
we can also use them to show type safety for the generalized form of
typability we will consider in Section~\ref{sec:genpresstlc}.

\subsection{Quasi-Stuck Terms}

We begin by inductively defining the set of \textbf{quasi-stuck} terms
$S$, in Figure~\ref{fig:qs}.  Also, let us call a quasi-stuck term
which is not a value \textbf{stuck}.  The purpose of these definitions
is to generalize a characterization of $c$-normal standard terms to
mixed terms (Lemmas~\ref{lem:qsnf} and~\ref{lem:qsnfb}, proved next),
in such a way that we can show that the set of quasi-stuck terms is
closed under abstract reduction (Lemma~\ref{lem:redqs}, proved below).
This will allow us to prove that typable quasi-stuck terms must be
values (Lemma~\ref{lem:stuck}), from which we easily obtain the
desired main theorems of progress and type safety.

\begin{figure}[b]
\begin{iteMize}{$\bullet$}
\item Mixed values $u$ are in $S$.
\item Terms of the form $a\ s$ or $A\ s$ are in $S$ if $s\in S$.
\item Terms of the form $f\ s$ or $(A\To A)\ s$ are in $S$ if $s\in S$ and $s$ is neither $a$ nor $A$.
\item Terms of the form $(\lambda x:T.m)\ s$ or $(T\To m)\ s$ are in $S$ if $s\in S$ and $s$ is not a mixed value.
\item Terms of the form $s\ s'$ are in $S$ if $s,s'\in S$ and $s$ is not a mixed value.
\end{iteMize}
\caption{Inductive definition of the set $S$ of quasi-stuck terms}
\label{fig:qs}
\end{figure}

\begin{lem}
\label{lem:qsnf}
If $m$ is quasi-stuck, then $m\not\to_c$.
\end{lem}
\begin{proof} The proof is by an easy structural induction on $m$, 
using the definition of quasi-stuck. 
\end{proof}

\begin{lem}
\label{lem:qsnfb}
If standard term $t$ is closed and $t\not\to_c$, then $t$ is
quasi-stuck.
\end{lem}
\begin{proof}
The proof is by structural induction on $t$.  If $t$ is
a (standard) value it is quasi-stuck, and it cannot be a variable
since $t$ is closed.  So suppose it is an application $t_1\ t_2$.
Since $t_1$ is closed, $t_1$ cannot be a variable.  We consider now
the remaining possibilities.  It could be that $t_1$ is $a$ and $t_2$
is some other $c$-normal form.  Then by the induction hypothesis,
$t_2$ is quasi-stuck, and $t$ is, too, using the second clause above
in the definition of quasi-stuck terms.  Next, we could have the
situation where $t_1$ is $f$, and $t_2$ is any $c$-normal form except
$a$.  Then by the induction hypothesis, $t_2$ is quasi-stuck, and $t$
is, too, using the third clause in the definition of quasi-stuck
terms.  Next, we could have that $t_1$ is a $\lambda$-abstraction, and
$t_2$ is any $c$-normal form except a standard value.  Then by the
induction hypothesis, $t_2$ is quasi-stuck, and it cannot be a mixed
value other than a standard value, because $t_2$ is a standard term.
So $t$ is quasi-stuck, too, using the fourth clause .  Finally, if
$t_1$ is some application, then by the induction hypothesis, $t_1$ and
$t_2$ are both quasi-stuck.  Since $t_1$ is not a value, the fifth
clause above gives us that $t$ is quasi-stuck. 
\end{proof}

\begin{lem}[Reduction of Quasi-Stuck Terms]
\label{lem:redqs}
If $m$ is quasi-stuck, and $m\to_a m'$, then $m'$ is also quasi-stuck.
Furthermore, if $m$ is a mixed value, then so is $m'$; and if $m$ is
not a mixed value, then neither is $m'$.
\end{lem}
\begin{proof}
The proof is by structural induction on
$m$.  Suppose $m$ is a mixed value.  Then it is easy to see by
inspection of the reduction rules that $m'$ must be, too.  So suppose
$m$ is of the form $a\ s$ or $A\ s$ with $s\in S$.  Then either the
assumed reduction is of the form $a\ s\to_a A\ s$, or else of the form
$a\ s\to_a a\ m''$ or $A\ s\to_a A\ m''$.  In the former case, the
resulting term is a quasi-stuck non-value.  In the latter, we may
apply the induction hypothesis to conclude that $m''$ is quasi-stuck,
and hence $a\ m''$ (or $A\ m''$) is a quasi-stuck non-value.

If $m$ is of the form $f\ s$ or $(A\To A)\ s$, where $s\in S$ and $s$
is not $a$ or $A$, then either the assumed reduction is of the form
$f\ s\to_a (A\To A)\ s$ or else $f\ s\to_a f\ m''$ or $(A\To
A)\ s\to_a (A\To A)\ m''$.  In the former case, the resulting term is
a quasi-stuck non-value, by the third clause of the definition of
quasi-stuck terms above.  In the latter, if $s$ is not a value, we
again use our induction hypothesis to conclude that $m''$ is a
quasi-stuck non-value, and hence not $a$ or $A$.  So $m'$ is a
quasi-stuck non-value, too.  If $s$ is a value, then so is $m''$, and
reduction cannot turn a value other than $a$ into $a$ or $A$.  So
again, $m''$ has the required form to be a quasi-stuck non-value.

Suppose $m$ is of the form $(\lambda x:T.m'')\ s$ or $(T\To m'')\ s$,
with $s\in S$ and $s$ not a mixed value.  Then either the assumed
reduction is of the form $(\lambda x:T.m'')\ s\to_a
(T\To[T/x]m'')\ s$; or else of the form $(\lambda x:T. m'')\ s\to_a
(\lambda x:T. m''')\ s$ or $(T\To m'')\ s\to_a (T\To m''')\ s$; or
else of the form $(\lambda x:T. m'')\ s\to_a (\lambda x:T. m'')\ m'''$
or $(T\To m'')\ s\to_a (T\To m'')\ m'''$.  In the first two cases, the
resulting term still has the required form to be a quasi-stuck
non-value.  In the third case, we know $s$ is not a value by the
definition of quasi-stuck terms, so we may use our induction
hypothesis to conclude that $m'''$ is a quasi-stuck non-value, which
is sufficient to conclude that the resulting term is again stuck.

Finally, suppose $m$ is of the form $m_1\ m_2$, where $m_1$ is not a
mixed value.  Then the assumed reduction must be of the form either
$m_1\ m_2\to_a m_1'\ m_2$ or else $m_1\ m_2\to_a m_1\ m_2'$, for some
$m_1'$ with $m_1\to_a m_1'$, or else some $m_2'$ with $m_2 \to_a
m_2'$.  This is because, by inspection of the reduction rules, $m$
itself cannot be a redex if $m_1$ is not a mixed value.  In the former
case, we may apply the induction hypothesis to conclude that $m_1'$
is a quasi-stuck non-value, and hence so is $m'$.  In the latter,
we may apply the induction hypothesis to conclude that $m_2'$ is quasi-stuck,
and hence so is $m'$. 
\end{proof}

\begin{lem}
\label{lem:cac}
If $m$ is quasi-stuck (including the case where $m$ is a
closed mixed value), and $m\to_{ca}^* T$, then $m \to_a^* T$.
\end{lem}
\begin{proof}  The proof is by induction on the length of
the reduction sequence from $m$ to $T$.  If this length is $0$, the
result obviously holds.  So suppose we have $m \to_{ca} m' \to_{ca}^*
T$.  Since $m$ is quasi-stuck, we have $m\not\to_c$ by
Lemma~\ref{lem:qsnf}.  So it must be the case that $m\to_a m'$.  Since
$m'$ is quasi-stuck by Lemma~\ref{lem:redqs}, we may apply our
induction hypothesis to conclude $m'\to_a^* T$, and hence $m\to_a^*
T$.
\end{proof}

\begin{lem}
\label{lem:stuck}
Suppose $m$ is a closed quasi-stuck term.  Suppose further that $m \to_{ca}^* T$.
Then $m$ is a mixed value.
\end{lem}
\begin{proof} The proof is similar to the previous one,
and proceeds by induction on the length of the reduction sequence from
$m$ to $T$.  If this length is $0$, the result holds, since types are
mixed values.  So suppose we have $m \to_{ca} m' \to_{ca}^* T$.  Since
$m$ is quasi-stuck, we have $m\not\to_c$ by Lemma~\ref{lem:qsnf}.  So
it must be the case that $m\to_a m'$.  We now consider cases on the
form of $m$.  If $m$ is a mixed value the result holds.  So suppose it
is a non-value.  Then by Lemma~\ref{lem:redqs}, $m'$ must also be a
quasi-stuck non-value, and we may apply the induction hypothesis to
derive a contradiction. \end{proof}

\subsection{Concluding Progress and Type Safety}

Armed with the concept of quasi-stuck terms and its associated lemmas,
we can now obtain the main results of this section.

\begin{thm}[Progress]
\label{thm:progstlc}
If standard term $t$ is closed, $t\to_a^* T$, and $t\not\to_c$, then
$t$ is a (standard) value.
\end{thm}

\begin{proof} By Lemma~\ref{lem:qsnfb} and the assumption $t\not\to_c$, we know $t$ is quasi-stuck.
Now since our assumption that $t\to_a^* T$ implies $t \to_{ca}^* T$,
we can apply Lemma~\ref{lem:stuck} to conclude that $t$ is a mixed
value (and hence a standard value, since $t$ is a standard term).  
\end{proof}

\begin{thm}[Type Safety]
\label{thm:safety}
If standard term $t$ is closed, $t\to_a^* T$, and $t\to_c^*
m\not\to_c$, then $m$ is a standard value.
\end{thm}
\begin{proof} The proof is by induction on the length of the reduction
sequence from $t$ to $m$.  In the base case, we apply
Theorems~\ref{thm:progstlc}, since we have $m = t\not\to_c$ in that
case.  For the step case, suppose we have $t\to_c m'\to_c^*
m\not\to_c$.  In this case, we can apply Theorem~\ref{thm:presstlc} to
conclude $m'\to_a^*T$.  It is easily proved by induction on the
structure of call-by-value evaluation contexts $E_c$ that if we have
$t\to_c m'$, then $m'$ is a standard term $t'$.  We may now apply the
induction hypothesis, since we have $t'\to_a^*T$ and $t'\to_c
m\not\to_c$. \end{proof}

\section{Applying Automated Analysis Tools to Type Preservation}
\label{sec:unistc}

In this section, we show how automated tools for analyzing
term-rewriting systems can be applied to automate part of the proof of
type preservation.  We will consider a language, which we call
Uniform-STC, that does not distinguish terms and types syntactically.
Advanced type systems like Pure Type Systems must often rely solely on
the typing rules to distinguish terms and types (and kinds,
superkinds, etc.)~\cite{B92}.  In Uniform-STC, we explore issues that
arise in applying the rewriting approach to more advanced type
systems.  We must now implement kinding (i.e., type checking of types)
as part of the abstract reduction relation.  We adopt a combinatory
formulation so that the abstract reduction relation can be described
by a first-order term-rewriting system.

\begin{figure}[t]
  \centering
\[
\begin{array}{lll}
  \textit{mixed terms}\ t\!&\!\!::=\!\!&\!\!S\langle
  t_1,t_2,t_3\rangle\ |\ K\langle t_1,t_2\rangle\ |\ t_1\ t_2
  |\ t_1\To t_2\ |\ A\ |\ \textit{kind}(t_1,t_2)\\
  \textit{mixed values}\ u\!&\!\!::=\!\!&\!\!S\langle t_1,t_2,t_3 \rangle\ |\ K\langle t_1,t_2\rangle\ |\ A\ |\ t_1\To t_2  \\
  \textit{concrete evaluation contexts}\ E_c\!&\!\!::=\!\!&\!\!*\ |\ E_c\ t\ |\ u\ E_c
\end{array}
\]
  \caption{Uniform-STLC language syntax and evaluation contexts}
  \label{fig:unified-syntax}
\end{figure}

\begin{figure}[t]
\[
  \begin{array}{ll}
  \textit{c}(\beta\textit{-S}). \!\!&\!\! \infer{E_c[S\langle t_1,t_2,t_3\rangle\ u\ u'\ u''] \to_c E_c[u\ u''\ (u'\ u'')]}{\ }
\\ \\
  \textit{c}(\beta\textit{-K}). \!\!&\!\!  \infer{E_c[K\langle t_1,t_2\rangle\ u\ u'] \to_c E_c[u]}{\ }
\\ \\
  \textit{a}(S). \!\!&\!\!  S\langle t_1,t_2,t_3\rangle \to_a \textit{kind}(t_1,\textit{kind}(t_2,\textit{kind}(t_3, (t_1\To t_2 \To t_3) \To (t_1\To t_2) \To (t_1 \To t_3))))
  \\
  \textit{a}(K). \!\!&\!\!  K\langle t_1,t_2\rangle \to_a \textit{kind}(t_1,\textit{kind}(t_2,(t_1\To t_2 \To t_1)))
  \\
  \textit{a}(\beta). \!\!&\!\!  (t_1\To t_2)\ t_1 \to_a \textit{kind}(t_1,t_2)
  \\
  \textit{a}(\textit{k-}\To). \!\!&\!\!  \textit{kind}((t_1 \To t_2), t) \to_a \textit{kind}(t_1,\textit{kind}(t_2,t))
  \\
  \textit{a}(\textit{k-A}). & \textit{kind}(A,t) \to_a t
  \end{array}
\]
  \caption{Concrete and abstract reduction rules}
  \label{fig:uni-rules}
\end{figure}

Figure~\ref{fig:unified-syntax} shows the syntax for the Uniform-STC
language.  There is a single syntactic category $t$ for mixed terms
and types, which include a base type $A$ and simple function types.
$S\langle t_1,t_2,t_3\rangle$ and $K\langle t_1,t_2\rangle$ are the
usual combinators, indexed by terms which determine their simple
types.  The \textit{kind} construct for terms is used to implement
kinding.  The rules for concrete and abstract reduction are given in
Figure~\ref{fig:uni-rules}.  The concrete rules are just the standard
ones for call-by-value reduction of combinator terms.  For abstraction
reduction, we are using first-order term-rewriting rules (unlike for
previous systems).  


For STLC (Section~\ref{sec:presstlc}), abstract $\beta$-redexes have the
form $(T\To t)\ T$.  For Uniform-STC, since there is no syntactic
distinction between terms and types, abstract $\beta$-redexes take the
form $(t_1\To t_2)\ t_1$, and we must use kinding to ensure that $t_1$
is a type.  This is why the $a(\beta)$ rule introduces a
\textit{kind}-term.  We also enforce kinding when abstracting simply
typed combinators $S\langle t_1,t_2,t_3\rangle$ and $K\langle
t_1,t_2\rangle$ to their types.  The rules for \textit{kind}-terms
($a(\textit{k-}\To)$ and $a(\textit{k-A})$) make sure that the first
term is a type, and then reduce to the second term.

Here, we define typability by value $u$ to mean abstract reduction to
$u$ where $u$ is \emph{kindable}, which we define as
$\textit{kind}(u,A)\to_a^* A$.  This definition avoids the need to
define types syntactically.

Following the methodology embodied in Theorem~\ref{thmtp}, we must
first prove the abstract reduction is confluent.  In fact, it is
convergent, and we can apply analysis tools to determine this, as
shown in the next two theorems.

\begin{thm}
  \label{uni-term}
  The term rewriting system $\to_a$ is terminating.
\end{thm}

\begin{proof}
The automated termination checker \textsc{Aprove} reports that the
rewrite system for $\to_a$ is terminating, using a recursive path
ordering~\cite{aprove}.
\end{proof}

\begin{thm}
\label{uni-confluent}
  The term rewriting system  $\to_a$ is confluent.
\end{thm}
\begin{proof}
Abstract reduction for Uniform-STC does not have the diamond property
due to the non-left-linear rule $a(\beta)$, where there could indeed
be redexes in the expressions matching the repeated variable $t_1$.
By Theorem~\ref{uni-term}, however, we can apply Newman's Lemma to
conclude confluence from local confluence.  Local confluence follows
because all the $aa$-peaks can be joined using either one $a$-step on
either side as for STLC, or else using additional balancing steps if
one of the rules applied is $a(\beta)$.  

But even easier than this reasoning is applying an automated
confluence checker: the ACP tool immediately reports that the abstract
reduction relation is confluent~\cite{aoto+09}.
\end{proof}

The proofs of Theorems~\ref{uni-term} and~\ref{uni-confluent}
demonstrate how the rewriting approach to typing benefits from recent
advances in analysis tools for term rewriting: we can use termination
and confluence checkers to analyze the abstract reduction relation
$\to_a$ corresponding to typing.  We expect this situation to recur
for more advanced type systems, although some may provide new
challenges for automated analysis tools (we give an example below).

\begin{lem}
\label{uni-complete}
  $\leftarrow_a \cdot \textit{Id}_{\leftarrow_a^*(S)} \cdot \to_c \subse (\to_c \cup \to_a^*) \cdot \leftarrow_a^*$. 
\end{lem}

\begin{proof}
We distinguish the peaks originating at typable terms $t$. 

If $\leftarrow_a$ and $\to_c$ steps are parallel --  $E'_c[t]
\leftarrow_a E_c[t] \leftarrow_{id} E_c[t] \to_{id} E_c[t] \to_c E_c[t']$ -- the peak can be completed
directly $E'_c[t] \to_{id} E'_c[t] \to_c E'_c[t'] \leftarrow_a E_c[t']
\leftarrow_{id} E_c[t']$.

If the $\leftarrow_a$ and $\to_c$ steps overlap, there are two cases,
corresponding to $\textit{c}(\beta\textit{-K})$ and
$\textit{c}(\beta\textit{-S})$ reduction steps. We show the completion
for $\textit{c}(\beta\textit{-K})$ peaks (omitting the $\to_{id}$
steps to simplify the presentation); the argument for
$\textit{c}(\beta\textit{-S})$ peaks is similar.

\[
\begin{array}{ll}
P. & E_c[u[(\hat{t}\ t\ t')]] \leftarrow_a E_c[(K\langle t_1,t_2\rangle\ t\ t')] \to_a E_c[t]\\
L. & E_c[u[(\hat{t}\ t\ t')]] \to_a^*  E_c[u[(\hat{t}\ t_1\ t'')]] \to_a E_c[u[((t_2\To t_1)\ t'')]] \to_a^*\\
\ & E_c[u[((t_2\To t_1)\ t_2)]] \to_a E_c[\textit{kind}(t_1,\textit{kind}(t_2,t_1))] \to_a^* E_c[t_1] \\
R. & E_c[t] \to_a^*  E_c[t_1]
\end{array}
\]

\noindent The $\to_a^*$-steps are justified because the peak term
(shown on line (P)) is typable by composition with $\textit{Id}_{\leftarrow_a^*(S)}$.  By confluence of
abstract reduction, this implies that the sources of all the left
steps are also typable.  For each $\to_a^*$-step, since abstract
reduction cannot drop redexes (as all rules are non-erasing), we argue
as for STLC that a descendant of the appropriate displayed
\textit{kind}-term or application must eventually be contracted, as
otherwise, a stuck descendant of such would remain in the final term.
Kindable terms cannot contain stuck applications or stuck
\textit{kind}-terms, because our abstract reduction rules are
non-erasing.  And contraction of those displayed \textit{kind}-terms
or applications requires the reductions used for the $\to_a^*$-steps,
which are sufficient to complete the peak.
\end{proof}

\begin{lem}
  \label{uni-normal}
  Every normal form with respect to $\to_a$ is also a normal form with respect to $\to_b$.
\end{lem}

The normal forms of $\to_a$ include $A$, $t_1 \To t_2$ where $t_1$
and $t_2$ are $a$-normal forms, $(t_1 \To t_2)\ t'_1$ where $t_1 \not= t'_1$, and
$kind(t_1,t)$ where $t_1$ is not generated by the grammar $T ::= A | T
\To T$. By inspection, $E_c[A] \not\to_c$ and $E_c[t_1 \To t_2] \not\to_c$.

\begin{thm}[Type Preservation]
\label{thm:presuni}
Let $m,m'$ be mixed terms and $T$ be a term such that $kind(T,t) \to_a t$. 
If $m \to_a^* T$ and $m\to_c m'$, then $m'\to_a^* T$.

\begin{proof}
  By application of Theorem~\ref{thmtp}. Condition (1) is satisfied by
  instantiating $S$ by the set  of terms $\{ t | kind(t,t') \to_a
  t\}$. Condition (2) follows by
  Theorem~\ref{uni-confluent}. Condition (3) by
  Lemma~\ref{uni-complete}, condition (4) by Lemma~\ref{uni-normal}.
\end{proof}
\end{thm}

\begin{thm}
  \label{uni-combined-confluent}
  Every mixed typable term is confluent with respect to the reduction
  relation $\to_{ac}$.

  \begin{proof}
For proving that $\to_{ba}$ is confluent for typable mixed terms we 
need to check properties (1) to (4) of Theorem \ref{thmcr} for the 
particular relations $\textit{Id}_{\leftarrow_a^*(S)} \cdot \to_a$
and $\textit{Id}_{\leftarrow_c^*(T)} \cdot \to_c$. The composition
of $\to_a$ and $\to_b$ with $\textit{Id}_{\leftarrow_a^*(S)}$ serves
to ensure that we are only considering typable terms.

Property (2) follows from Theorem \ref{uni-confluent} and the fact that
$\textit{Id}_{\leftarrow_a^*(\textit{Types})}$ is the identity relation. All 1-step
peaks of must be of the form $m \leftarrow m \to m$, due to the
composition with $\textit{Id}_{\leftarrow_a^*(\textit{Types})}$. By Theorem
\ref{uni-confluent}, if $m_1 \leftarrow_a m \to_a m_2$, then there exists
$m_3$ such that $m_1 \to_a^* m_3 \leftarrow_a m_2$. Thus, any
$\textit{Id}_{\leftarrow_a^*(\textit{Types})} \cdot \to_a$ peak $m_1 \leftarrow_a m
\leftarrow_{id} m \to_{id} m \to_a m_2$ can be completed with $m_1
\to_{id} m_1 \to_a m_3 \leftarrow_a m_2 \leftarrow_{id}
m_2$. By Theorem \ref{uni-term} $\to_a$ is terminating,
so $\textit{Id}_{\leftarrow_a^*(\textit{Types})} \cdot \to_a \subse \to_a$ is also
terminating, proving property (1). Property (3) follows from Lemma
\ref{uni-complete}. Property (4) follows from Lemma~\ref{uni-normal}.
  \end{proof}
  
\end{thm}

As an aside, note that a natural
modification of this problem is out of the range of ACP, version 0.20.
Suppose we are trying to group kind-checking terms so that we can
avoid duplicate kind checks for the same term.  For this, we may wish
to permute \textit{kind}-terms, and pull them out of other term
constructs.  The following rules implement this idea, and can be
neither proved confluent nor disproved by ACP, version 0.20.  Just the
first seven rules are also unsolvable by ACP. {\small
\begin{verbatim}
(VAR a b c A B C D)
(RULES
  S(A,B,C) -> kind(A,kind(B,kind(C,
              arrow(arrow(arrow(A,arrow(B,C)),arrow(A,B)),arrow(A,C)))))
  K(A,B) -> kind(A,kind(B,arrow(A,arrow(B,A))))
  app(arrow(A,b),A) -> kind(A,b)
  kind(base,a) -> a
  kind(arrow(A,B),a) -> kind(A, kind(B, a))
  kind(A,kind(A,a)) -> kind(A,a)
  kind(A,kind(B,a)) -> kind(B,kind(A,a))
  app(kind(A,b),c) -> kind(A,app(b,c))
  app(c,kind(A,b)) -> kind(A,app(c,b))
  arrow(kind(A,b),c) -> kind(A,arrow(b,c))
  arrow(c,kind(A,b)) -> kind(A,arrow(c,b))
  kind(kind(a,b),c) -> kind(a,kind(b,c))
)
\end{verbatim}
}

\section{Generalizing Nuprl's Direct Computation Rules}
\label{sec:genpresstlc}

Martin-L\"of's Intuitionistic Type Theory (ITT), as formulated
in~\cite{martinloef+84}, is a system of four judgments presented with
a rigorous but informal semantics.  A typing judgment of the form
$a\in A$ ``means that $a$ has a canonical object of the canonical type
denoted by A as value''~\cite[page 174]{martinloef+84}.  Here,
Martin-L\"of is making use of the concept of a term (of ITT) having a
value, a concept he defines earlier in the paper.  The authors of the
Nuprl system realized that this semantics justifies more permissive
typing rules than allowed by Martin-L\"of's own formal
systems~\cite{constable+86} (see also Section 2.2 of~\cite{allen+06}
for a historical perspective).  In particular, it justifies so-called
\emph{direct computation} rules, which turned out to be useful for
formal development with Nuprl:
\[
\infer{t\in T}{t\to^*t' & t'\in T}
\]
\noindent Applying Theorem~\ref{thm:relatetyp}, we can view this rule
from a rewriting perspective.  We will use call-by-value reduction, as
full $\beta$-reduction would require additional technicalities that
would not be illuminating (we would have to use parallel reduction and
incorporate a proof of confluence of $\beta$-reduction, in order to
get preservation of generalized typing).
\[
\infer{t\to_a^* T}{t\to_c^*t' & t'\to_a^* T}
\]
\noindent In this section, we will take the idea of Nuprl's direct
computation rules one step further, by adopting the following
definition.

\begin{defi}[Generalized Typability]
A mixed term $m$ is called \textbf{generalized typable}
if $m\to_{ca}^* T$ for some type $T$.
\end{defi}

\noindent This allows us to view (call-by-value versions of) Nuprl's
direct computation rules as embodying a special case of generalized
typability, namely $\to_c^*\cdot\to_a^*$.  We will see in this section
that we can prove type preservation directly for generalized typing,
using the rewriting approach.  Note that generalized typability is not
obviously decidable, since $\to_{ca}$ is not terminating

A simple example of generalized typability is given by the
term $(\lambda x:A.\ \lambda y:A.y)\ \lambda x:A. x\ x$.  Note that
the argument term $\lambda x:A. x\ x$ is not simply typable.  This
term has several $ca$-reduction sequences, including the following
one:
\[
\begin{array}{l}
(\lambda x:A.\ \lambda y:A.y)\ \lambda x:A.x\ x\ \to_a\\
(\lambda x:A.\ (A \To A))\ \lambda x:A.x\ x\ \to_a \\
(\lambda x:A.\ (A \To A))\ (A\To (A\ A))\ \to_c \\
A\To A
\end{array}
\]
\noindent Because this term $ca$-reduces to a type, the generalized
type-safety property we will obtain in this section tells us that the
$c$-normal form of this term, if such exists, is a value.  This can,
of course, be confirmed for this case, where the $c$-normal form is
just $\lambda y:A.y$.  Notice that this example also shows that
$\to_{ca}$ is not confluent, as we can also reduce it to a stuck term
in this way:
\[
\begin{array}{l}
(\lambda x:A.\ \lambda y:A.y)\ \lambda x:A.x\ x\ \to_a\\
(\lambda x:A.\ (A \To A))\ \lambda x:A.x\ x\ \to_a \\
(\lambda x:A.\ (A \To A))\ (A\To (A\ A))\ \to_a \\
(A\To (A \To A))\ (A\To (A\ A))\ \not\to_{ca}
\end{array}
\]

\begin{thm}[Generalized Type Preservation for Call-By-Value Reduction]
\label{thm:genpresstlc}
If $m \to_{ca}^* T$ and $m\to_c m'$, then $m'\to_{ca}^* T$.
\end{thm}

\begin{proof}
We cannot conveniently apply Theorem~\ref{thmtp}, because the natural
instantiation would be to take $\to_{ca}$ for the relation $\to_a$ in
the theorem -- but then we would have to prove confluence of
$\to_{ca}$, which does not hold (as shown just above).  So instead we
give a direct proof, by induction on the length of the assumed
$ac$-sequence from $m$ to $T$.  The sequence cannot be of length $0$,
since $m$ cannot be a type (since it $c$-reduces, as no type can).

For the step case: suppose the assumed $ca$-reduction is of the form
$m \to_a m''\to_{ca}^* T$.  We now consider cases for the form of
overlap of the step $m\to_a m''$ and $m\to_c m'$.  Suppose the
$c$-step is $E_c[f\ a]\to_c E_c[a]$.  If the $a$-step is in $E_c$,
that means $m'' = E_c'[f\ a]$, where the hole in $E_c$ is at the same
position as in $E_c'$.  We can just permute these steps, to obtain
$E_c[a]\to_a E_c'[a]$ and $E_c'[f\ a]\to_c E_c'[a]$.  Now the
induction hypothesis can be applied with $E_c'[f\ a]$ (i.e., $m''$) as
the peak term, and $E_c'[a]$ as the term to which it $c$-steps.  

So suppose the $a$-step is in the displayed $f\ a$ of $E_c[f\ a]$.
Then before the reduction sequence from $m''$ to $T$ can perform a
$c$-step, it must first reduce the residual of $f\ a$ to $A$, since
that residual occurs in a $c$-reduction position.  So the reduction
sequence from $m''$ to $T$ must look like the following, where the
hole in $E_c$ and in $E_c'$ are at the same position:
\[
m'' \to_a^* E_c'[A] \to_{ca}^* T
\]
\noindent By performing the $a$-reductions which transformed $E_c$ to
$E_c'$, we can reduce $E_c[a]$ to $E_c'[A]$, and then we are done,
since we then have $m'\to_a^* E_c'[A]\to_{ca}^* T$.

We now must consider the case where the $c$-step is $E_c[(\lambda
  x:T'.m_1)\ u]\to_c E_c[[u/x]m_1]$.  Again, if the $a$-step is in
$E_c$, we can permute steps and apply the induction hypothesis.  If
the $a$-step is in $m_1$ or in $u$, we can also permute the steps,
though if the reduction is in $u$ (say $u\to_a u'$), we will in
general have $E_c[[u/x]m_1]\to_a^* E_c[[u'/x]m_1]$, since $x$ need not
appear exactly once in $m_1$.  Nevertheless, we can still apply the
induction hypothesis with $m''$ as the peak term, since we will only
ever produce one $c$-step from $m''$ by permuting steps.  Finally,
suppose the $a$-step is $E_c[(\lambda x:T'.m_1)\ u]\to_a E_c[(T'\To
  [T'/x]m_1)\ u]$.  By similar reasoning as in the previous case, the
$ca$-reduction sequence from $E_c[(T'\To [T'/x]m_1)\ u]$ to $T$ may
contain $a$-steps transforming $E_c$ to some $E_c'$, but it cannot
take a $c$-step until it has reduced the displayed $(T'\To
[T'/x]m_1)\ u$ to $[T'/x]m_1'$, with $u\to_a^* T'$ and $m_1\to_a^*
m_1'$.  This is because that displayed term is in $c$-reduction
position and neither a value nor a redex.  We can then duplicate any
$a$-steps taken in $E_c$ to $a$-reduce $E_c[[u/x]m_1]$ (i.e., $m'$) to
$E_c'[[T'/x]m_1']$.  This term then $ac$-reduces to $T$, and we are
done.
\end{proof}

\begin{thm}[Generalized Progress]
\label{thm:genprogstlc}
If standard term $t$ is closed, $t\to_{ca}^* T$, and $t\not\to_c$,
then $t$ is a (standard) value.
\end{thm}

\begin{proof} As for Theorem~\ref{thm:progstlc}, we obtain
this result by applying Lemmas~\ref{lem:qsnfb} and~\ref{lem:stuck}. 
\end{proof}

\begin{thm}[Generalized Type Safety]
\label{thm:gensafety}
If standard term $t$ is closed, $t\to_{ca}^* T$, and $t\to_c^*
t'\not\to_c$, then $t'$ is a (standard) value.
\end{thm}
\begin{proof} This is a direct corollary of
Theorems~\ref{thm:genpresstlc} and~\ref{thm:genprogstlc}. \end{proof}

\section{A Rewriting Approach to Normalization for STLC}
\label{sec:normstlc}

In this Section, we will see how the rewriting approach to typing
impacts a standard approach to proving that every typable (closed)
standard term of the simply typed lambda calculus has a $b$-normal
form.  We will work with a slightly different presentation of STLC
than we saw in Section~\ref{sec:restlc}, in particular dispensing with
the term constants $a$ and $f$.  We assume a non-empty set of type
constants $A$.  The syntax we are using in this section is:

\[
\begin{array}{lll}
\textit{types}\ T & ::= & A\ |\ T_1\To T_2\\
\textit{mixed terms}\ m & ::= & x\ |\ \lambda x:T.\,m\ |\ m\ m'\ | \ A\ |\ T\To m \\
\textit{standard terms}\ t & ::= & x\ |\ \lambda x:T.\,t\ |\ t\ t'
\end{array}
\]

\noindent The abstract and concrete reduction relations are then
defined as follows, where we use mixed terms $m$ as contexts
(sometimes using meta-variable $\hat{m}$ in this case), writing
$m[m']$ to denote the replacement of the unique occurrence of a
special variable $*$ in $m$ by $m'$.  

\[
\begin{array}{lll}
\infer[\textit{b}(\beta)]{\hat{m}[(\lambda x:T.\,m)\ m']\ \to_b\ \hat{m}[[m'/x]m]}{\ }
\\\\
\infer[\textit{a}(\beta)]{\hat{m}[(T \To m)\ T]\ \to_a\ \hat{m}[m]}{\ }
\\\\
\infer[\textit{a}(\lambda)]{\hat{m}[\lambda x:T.\, m]\ \to_a\ \hat{m}[T\To [T/x]m]}{\ }
\end{array}
\]

\subsection{Interpretation of Mixed Terms}
\label{sec:interp}

The proof in this section is based on ideas from standard proofs, such
as Girard's proof in the book \emph{Proofs and
  Types}~\cite{girard-proofs-types}.  The technical details evolve
differently, however, since we are using the rewriting approach to
typing.  Similarly to Girard's proof, we are going to define an
interpretation of open types as sets of standard terms.  Here, we need
to generalize this to give interpretations $\interp{m}_\phi$ of mixed
terms $m$, where (as standard) $\phi$ assigns interpretations to the
free variables of $m$.  The most enlightening observation that will
come from this is Theorem~\ref{thm:abstr} (Abstraction Theorem), which
says that interpretation is monotonic with respect to abstract
reduction: if $m\to_a m'$, then
$\interp{m}_\phi\subseteq\interp{m'}_\phi$.  If one views a set as
abstracting its elements, and if one considers a mixed term as a code
for the set of terms which is its interpretation, then the Abstraction
Theorem shows that more abstract codes have more abstract
interpretations.  This is an elegant perspective that arises -- from
the standard Tait-Girard method -- only by taking a small-step view of
typing; existing proofs for normalization in the literature do not
have any theorem which corresponds (in any obvious way) to the
Abstraction Theorem.

So now to begin the development, let WN be the set of standard terms
which are weakly normalizing with respect to $\to_b$ (that is, terms
$t$ such that there exists some $t'$ such that $t \to_b^* t'
\not\to_b$).  Also, if $\to$ is any binary relation on standard terms
and $R$ any set of standard terms, we will write $\to(R)$ for the
image of $R$ under $\to$ (that is, $\{ t'\ |\ \exists t\in R.\ t\to
t'\}$). 

\

\noindent We first define $\mathcal{R}$ to be the set of all sets $R$
of standard terms satisfying the following conditions:
\begin{enumerate}[(1)]
\item $\ot_b^*(R)\ \subseteq\ R$
\item $R \neq \emptyset$
\item $R\subseteq WN$
\end{enumerate}
\noindent The first condition ensures that $t'\to_b^* t$ and $t\in R$
imply $t'\in R$.  An assumption like this is often made about such
sets of terms.  We will call elements of $\mathcal{R}$
\emph{reducibility sets}.  Much work has been devoted to comparing
different conditions for families of sets in the context of the
interpretation of types (see, e.g.,~\cite{riba07,gallier90}).  Our
focus here is not so much on the specific conditions on the
interpretations of mixed terms, as on how interpretations of terms in
the abstract reduction relation are related.  The conditions we adopt
here are simple and sufficient for weak normalization of closed terms
(cf. also Chapter 12 of~\cite{pierce02}).

We will use $\phi$ as a meta-variable for \emph{assignments}, which are
functions from \textit{Var} to $\mathcal{R}$.  We write $\phi[R/x]$ to
mean the function $\phi$ updated to map variable $x$ to
$R\in\mathcal{R}$.  Now for any $m$ and $\phi$ with
$\textit{FV}(m)\subseteq\textit{dom}(\phi)$, we define the
interpretation $\interp{m}_\phi$ of $m$ with respect to $\phi$ in
Figure~\ref{fig:interp}.  To ensure that interpretations of types
satisfy the first property above of reducibility sets, we need
to close under $\ot_b^*$ in the last two clauses of the definition
(in Figure~\ref{fig:interp}).  Since we are proving normalization, we
take the set of normalizing terms as the interpretation of $A$,
similarly to what is standardly done for atomic types (e.g., in
Girard's proof).

\begin{figure}
\[
\begin{array}{lll}
\interp{T \To m}_\phi & = & \{ t\ |\ \forall t'\in\interp{T}_\phi.\ t\ t'\in\interp{m}_\phi \} \\
\interp{x}_\phi & = & \phi(x) \\
\interp{A}_\phi & = & \textnormal{WN} \\
\interp{\lambda x : T.m}_\phi & = & \ot_b^* (\{ \lambda x : T. t\ |\ \forall t'\in\interp{T}_\phi.\ [t'/x]t\in\interp{m}_{\phi[\interp{T}_\phi/x]} \}) \\
\interp{m_1\ m_2}_\phi & = & \ot_b^*(\{ t_1\ t_2\ |\ t_1\in\interp{m_1}_\phi\ \wedge\ t_2\in\interp{m_2}_\phi \}
\end{array}
\]
\caption{The interpretation of mixed terms}
\label{fig:interp}
\end{figure}

\subsection{Interpretations of Types are Reducibility Sets}

In this section, we prove that for all types $T$ and $\phi$ with
$\textit{FV}(T)\subseteq\textit{dom}(\phi)$, we have
$\interp{T}_\phi\in\mathcal{R}$.  We will elide this condition
relating $T$ (or instead $m$) and $\phi$ below.  We prove the three
properties of reducibility sets given in the previous section.
The properties must be proved in order, as later properties depend on
earlier ones.  The first property is needed in a more general form,
for any mixed term $m$, and not just types $T$.  The second two
properties are only needed for types.  The proofs in this section are
similar to those used for the standard definition of typing, except
that there, they are usually proved by mutual induction.  Here we can
prove them independently, though in sequence, due to the simpler form
of the second property.  While the development in this section is
similar to the usual one, in the next section we will see something
significantly different.

\begin{lem}
\label{lem1}
$\ot_b^*\interp{m}_\phi\subseteq \interp{m}_\phi$
\end{lem}
\begin{proof} The proof is by structural induction on $m$.  If $m$ is a $\lambda$-abstraction,
or application, the desired property follows by idempotence of
$\ot_b^*$ as an operator on sets of terms.  If $m$ is a variable $x$,
then the property follows by the same property for $\phi(x)$, since we
stipulated assignments map variables to elements of $\mathcal{R}$.  If
$\phi = A$, then we must prove
\[
\ot_b^*(\textit{WN})\subseteq \textnormal{WN}
\]
\noindent But this just amounts to the obvious fact that if $t'\to_b^*
t$ and $t$ is weakly normalizing, then $t'$ is also weakly
normalizing.

\

\noindent Finally, suppose $m$ is $T\To m'$ for some $m'$.  Assume an
arbitrary $t\in\interp{T\To m'}_\phi$, and arbitrary $t'$ with
$t'\to_b^* t$.  We must show $t'\in\interp{T\To m'}_\phi$.  To do
this, by the definition of the interpretation of $\To$-terms, it
suffices to consider arbitrary $t''\in\interp{T}_\phi$, and show
$t'\ t''\in\interp{m'}_\phi$.  We have $t\ t''\in\interp{m'}_\phi$ by
the definition of the interpretation of $\To$-terms.  Then we get the
desired conclusion by the induction hypothesis on $m'$, since
$t\ t''\to_b^* t'\ t''$.
\end{proof}

\begin{lem}
\label{lem2}
$\interp{T}_\phi \neq \emptyset$
\end{lem}

\begin{proof} The proof is by structural induction on $T$.  If $T$ is $A$, then the
desired property holds immediately, since $x$ is in
$\textnormal{WN} = \interp{A}_\phi$.  So suppose $T\equiv T_1\To T_2$, for some $T_1$
and $T_2$.  We must exhibit some $t\in\interp{T_1\To T_2}_\phi$.  By
the induction hypothesis applied to $T_2$, there exists some
$t'\in\interp{T_2}_\phi$.  Now take $\lambda x:T_1.t'$ for the
required term $t$, where we assume $x\not\in\textit{FV}(t')$.  We just have to confirm that $\lambda
x:T_1.t'\in\interp{T_1\To T_2}_\phi$.  So assume arbitrary
$t''\in\interp{T_1}_\phi$, and show $(\lambda
x:T_1.t')\ t''\in\interp{T_2}_\phi$.  By Lemma~\ref{lem1}, it suffices
to prove $t'\in\interp{T_2}_\phi$, since $(\lambda
x:T_1.t')\ t''\to_b^* t'$.  But we are assuming
$t'\in\interp{T_2}_\phi$.
\end{proof}

\begin{lem}
\label{lem3}
$\interp{T}_\phi \subseteq \textnormal{WN}$
\end{lem}

\begin{proof} The proof is again by structural induction on $T$, and is trivial when
$T$ is $A$.  So suppose $T\equiv T_1\To T_2$, and assume arbitrary
$t\in\interp{T_1\To T_2}_\phi$.  We must show $t\in\textnormal{WN}$.
By Lemma~\ref{lem2}, we know there exists some term
$t'\in\interp{T_1}_\phi$.  Then by the definition of the
interpretation of $\To$-terms, $t\ t'\in\interp{T_2}_\phi$.  By the
induction hypothesis applied to $T_2$, we then have
$t\ t'\in\textnormal{WN}$.  But this implies $t\in\textnormal{WN}$, as
required.
\end{proof}

\begin{cor}
$\interp{T}_\phi \in \mathcal{R}$
\end{cor}

\noindent The above lemmas have proved that $\interp{T}_\phi$
satisfies the three properties for membership in $\mathcal{R}$.  In
the next section, we will also need the following lemma, whose proof
is routine and omitted:

\begin{lem}[Semantic Substitution]
\label{lem:semsubst}
$\interp{[T/x]m}_\phi\ =\ \interp{m}_{\phi[\interp{T}_\phi/x]}$
\end{lem}

\subsection{The Abstraction Theorem}
\label{sec:abstr}

In this section, we prove a remarkable theorem, from which the
normalization property for typable terms will follow as a corollary.
For any mixed terms $m$ and $m'$, and any $\phi$ with
$\textit{FV}(m)\subseteq\textit{dom}(\phi)$, we have:

\begin{thm}[Abstraction Theorem]
\label{thm:abstr}
$m\to_a m'\ \Longrightarrow\ \interp{m}_\phi\subseteq\interp{m'}_\phi$
\end{thm}
\noindent Note that well-definedness of $\interp{m'}_\phi$ in the
statement of the theorem follows from the assumption about $\phi$ and
the observation that abstract reduction cannot introduce new
variables.

\

\noindent This theorem is remarkable because it reflects the essence of
abstraction: the gathering of different concrete entities under the
same abstract one.  The Abstraction Theorem shows that abstract
reduction is increasing the set of concrete terms which are collected
under a mixed (and so partially abstract) term.  In the next section,
we will see how to conclude normalization from this theorem.

\

\begin{proof}[Proof of Theorem~\ref{thm:abstr}]
It suffices to prove by structural induction on $\hat{m}$ that for all
$\phi$ and for all $m$ and $m'$ where $m$ is a redex and $m'$ its
contractum:
\[
\hat{m}[m] \to_a \hat{m}[m']\longrightarrow \interp{\hat{m}[m]}_\phi\subseteq\interp{\hat{m}[m']}_\phi
\]
\noindent \textbf{Case:} $\hat{m}\equiv m_1\ m_2$, where the hole is
in $m_1$.  The case where the hole is in $m_2$ is similar, so we omit
it.  To show the required
$\interp{m_1[m]\ m_2}_\phi\subseteq\interp{m_1[m']\ m_2}_\phi$,
consider arbitrary $t\in\interp{m_1[m]\ m_2}_\phi$.  By the definition
of the interpretation of applications, we must have
$t_1\in\interp{m_1[m]}_\phi$ and $t_2\in\interp{m_2}_\phi$ with
$t\to_b^* t_1\ t_2$.  Now by the induction hypothesis applied to
$m_1$ we have:
\[
\interp{m_1[m]}_\phi\subseteq\interp{m_1[m']}_\phi
\]
\noindent This implies $t_1\ t_2\in\interp{m_1[m']\ m_2}_\phi$.  From
this, we obtain the desired $t\in\interp{m_1[m']\ m_2}_\phi$ by the
definition of the interpretation of applications.

\

\noindent \textbf{Case:} $\hat{m}\equiv \lambda x:T.m_1$, for some
$x$, $T$, and $m_1$, with the hole in $m_1$.  Consider an arbitrary
$t\in\interp{\lambda x:T.m_1[m]}_\phi$.  By the definition of the
interpretation of $\lambda$-abstractions, this implies that there
exists a term $t_1$ such that $t\to_b^*\lambda x:T.t_1$ and for all
$t''\in\interp{T}_\phi$, we have
$[t''/x]t_1\in\interp{m_1[m]}_{\phi[\interp{T}_\phi/x]}$.  We must
show $t\in\interp{\lambda x:T.m_1[m']}_\phi$.  By the definition of
the interpretation of $\lambda$-terms and Lemma~\ref{lem1}, it suffices to
prove $(\lambda
x:T.t_1)\ t''\in\interp{m_1[m']}_{\phi[\interp{T}_\phi/x]}$ for
arbitrary $t''\in\interp{T}_\phi$.  Again applying Lemma~\ref{lem1},
we can see it suffices to prove
$[t''/x]t_1\in\interp{m_1[m']}_{\phi[\interp{T}_\phi/x]}$.  This now
follows by the induction hypothesis applied to context $m_1$.

\

\noindent \textbf{Case:} $\hat{m}=*$.  Now we must distinguish the
two cases for an abstract reduction.

\

\noindent \textbf{Case 1.} Suppose that we have
\[
\lambda x:T.m\ \to_a\ T\To[T/x]m
\]
\noindent We must prove $\interp{\lambda
  x:T.m}_\phi\subseteq\ \interp{T\To[T/x]m}_\phi$.  So assume
arbitrary $t\in\interp{\lambda x:T.m}_\phi$, and show
$t\in\interp{T\To[T/x]m}_\phi$.  To show that, it suffices to consider
arbitrary $t''\in\interp{T}_\phi$, and prove
$t\ t''\in\interp{[T/x]m}_\phi$.  By the definition of the
interpretation of $\lambda$-abstractions, we have $t\to_b^*\lambda
x:T.t'$, for some $t'$, with
$[t''/x]t'\in\interp{m}_{\phi[\interp{T}_\phi/x]}$ for all
$t''\in\interp{T}_\phi$.  Since $t\ t''\to_b^*[t''/x]t'$, it suffices
by Lemma~\ref{lem1} just to prove $[t''/x]t'\in\interp{[T/x]m}_\phi$.
This follows from the fact just derived, applying also Lemma~\ref{lem:semsubst}.

\

\noindent \textbf{Case 2.} Suppose that we have
\[
(T\To m)\ T\ \to_a\ m
\]
\noindent Assume an arbitrary $t\in\interp{(T\To m)\ T}_\phi$.  By the
definition of the interpretation of applications, we then have that
there exists $t_1\in\interp{T\To m}_\phi$ and $t_2\in\interp{T}_\phi$
such that $t\to_b^* t_1\ t_2$.  We must show $t\in\interp{m}_\phi$.
By the definition of the interpretation of $\To$-terms, we obtain
$t_1\ t_2\in\interp{m}_\phi$.  By Lemma~\ref{lem1}, this suffices
to  establish $t\in\interp{m}_\phi$, since $t\to_b^* t_1\ t_2$.
\end{proof}

\subsection{Concluding Normalization}

Using the Abstraction Theorem, we can obtain the main result that
typable terms are normalizing.  First, we need this helper lemma
stating that standard terms are in their own interpretations:

\begin{lem}
\label{lem:inown}
Consider an arbitrary standard term $t$ and assignment $\phi$, as well
as function $\sigma$ from variables to standard terms.  Suppose also
that for all $x\in\textit{FV}(t)$, we have $\sigma(x)\in\phi(x)$.  Then
we have $\sigma t\in\interp{t}_\phi$.
\end{lem}
\begin{proof}
The proof is by structural induction on $t$.  If $t$ is a
variable $x$, then we have $\sigma x\in\phi(x)$ by assumption.  If $t$
is of the form $\lambda x:T.t_1$, then the definition of the
interpretation of mixed terms tells us:
\[
\interp{\lambda x:T.t_1}_\phi = \ot_b^*(\{ \lambda x:T.t'\ |\ \forall t''\in\interp{T}_\phi.\ [t''/x]t'\in\interp{t_1}_\phi \})
\]
\noindent To show that $\sigma\lambda x:T.t_1$ is itself a member of the set
on the right-hand side of this equation, it suffices to consider an
arbitrary $t''\in\interp{T}_\phi$, and show $[t''/x](\sigma
t_1)\in\interp{t_1}_{\phi[\interp{T}_\phi/x]}$.  Here we can apply the
induction hypothesis for $t_1$, with $\sigma[t''/x]$ and
$\phi[\interp{T}_\phi/x]$.  The two substitutions still satisfy the
required properties.  Finally, if $t$ is of the form $t_1\ t_2$, the
result easily follows from the induction hypothesis applied to $t_1$
and also to $t_2$, and the definition of the interpretation of
applications.
\end{proof}

\begin{thm}[Normalization for Typable Terms]
For all closed standard terms $t$ and types $T$, if $t\to_a^* T$, then $t\in\textnormal{WN}$.
\end{thm}

\proof By Lemma~\ref{lem:inown}, we have
$t\in\interp{t}_\emptyset$.  Then by iterated application of
Theorem~\ref{thm:abstr}, we know that
$\interp{t}_\emptyset\subseteq\interp{T}_\emptyset$.  By
Lemma~\ref{lem3}, $\interp{T}_\emptyset\subseteq\textnormal{WN}$.
Putting these facts together, we get this chain of relationships,
which suffices:
\[
t \in\ \interp{t}_\emptyset\ \subseteq \ \interp{T}_\emptyset\ \subseteq\ \textnormal{WN}\eqno{\qEd}
\]

\subsection{Summary of The Standard Proof}

\newcommand{\redd}[1]{\textit{Red}_{#1}}
\newcommand{\nxt}[0]{\textit{next}}
\newcommand{\sn}[0]{\textit{SN}}

Here, we summarize Girard's proof of strong normalization, for
purposes of comparison~\cite{girard-proofs-types}.  This proof is
based on the usual judgment $\Gamma \vdash t : T$ for STLC.  One first
defines an interpretation of types:
\[
\begin{array}{lll}
t \in \redd{b} & \Leftrightarrow & t \in \sn \\
t \in \redd{T \to T'} & \Leftrightarrow & \forall t' \in \redd{T}.\ (t\ t') \in \redd{T'}
\end{array}
\]
\noindent This does not require use of a function $\phi$ as above
(though the standard proof for System F does).  For this
interpretation of types, one then proves these three properties, by
mutual structural induction on the type $T$ mentioned in all three
properties:

\begin{enumerate}[(1)]
\item $\redd{T}(t) \ \To\ \sn(t)$.
\item $\redd{T}(t)\ \To\ \redd{T}(\nxt(t))$.
\item If $t$ is neutral, then $\redd{T}(\nxt(t)) \ \To \  \redd{T}(t)$.
\end{enumerate}

\noindent A term is neutral iff it is not a $\lambda$-abstraction.
The third property implies that all the variables are in $\redd{T}$ for
every $T$.  Finally, one derives the following different theorem in
place of the Abstraction Theorem:

\begin{thm}[Reducibility]
\label{thm:red}
Suppose $\{ x_1 : T_1, \ldots, x_n : T_n \} \vdash t : T$, and
consider arbitrary $t_i\in\redd{T_i}$, for all $i\in\{1,\ldots,n\}$.
Then $[t_1/x_1, \ldots, t_n/x_n]\, t\in\redd{T}$.
\end{thm}

\noindent Now we can obtain as a corollary that $\Gamma \vdash t : T$
implies $t\in\sn$, since $\redd{b}\subseteq \sn$ by the first property
above, and a substitution $\sigma$ replacing $x$ by $x$ satisfies the
required condition, since all variables are included in all sets
$\redd{T}$.

\subsection{Discussion}

The main difference in the rewriting-based development and the
standard one is in deriving the Abstraction Theorem.  The form of the
theorem is completely different from Theorem~\ref{thm:red}.  One nice
technical feature is that for the proof of the Abstraction Theorem, we
did not need to apply a substitution to terms inhabiting
interpretations of types, as we did for Theorem~\ref{thm:red}.  We
still needed to use the idea of such a substitution, but it appeared
only in a simple helper lemma, namely Lemma~\ref{lem:inown}.  This is
an advantage of the rewriting-based version, since the substitution
does not clutter up the proof of the central result.  One disadvantage
of the rewriting-based version is that we needed the function $\phi$
and Lemma~\ref{lem:semsubst} -- but this is not such a significant
disadvantage, since those devices are needed when we move to System F
in the standard development anyway.

\section{Conclusion}
\label{sec:conclusion}

We have seen how rewriting techniques can be used to develop the
meta-theory of simple types.  Typing is treated as a small-step
abstract reduction relation, and type safety, based on type
preservation and progress theorems, can be established by analysis of
the interactions between abstract and concrete reduction steps.  A
crucial ingredient of our approach to type preservation, as defined by
Theorem~\ref{thmtp}, is to have a confluent abstract reduction
relation.  For simply typed lambda calculus, this was a trivial
matter, but we saw a more complex example, where applying automated
confluence-checking tools developed in the term-rewriting community
was able to automate this part of the type preservation proof.
Confluence of the combination of abstract and concrete reduction for
typable terms is an easy corollary of type preservation
(Theorem~\ref{thmcr}).  We have also seen how to adapt a standard
proof of normalization for simply typed terms, for the rewriting
approach to typing.  For this proof, mixed terms are interpreted as
sets of standard terms, and the crucial insight is embodied in the
Abstraction Theorem, which shows that those sets are enlarged by
reduction of the corresponding mixed terms.

There are many avenues for future work.  First, the rewriting approach
should be applied to more advanced type systems, including ones with
impredicative polymorphism.  Dependent type systems pose a particular
challenge, because from the point of view of abstract reduction,
$\Pi$-bound variables must play a dual role.  When computing a
dependent function type $\Pi x:T.\ T'$ from an abstraction $\lambda x
: T. t$, we may need to abstract $x$ to $T$, as for STLC; but we may
also need to leave it unabstracted, since with dependent types, $x$ is
allowed to appear in the range type $T'$.  It would also be
interesting to see if there are consequences of the rewriting approach
to typing when applied to proofs via the Curry-Howard isomorphism.
Theorem~\ref{thm:abstr} (Abstraction) shows how the set of proofs in
the meaning of a mixed proof term (part proof and part formula)
increases as the term is abstracted.  Certainly, the present methods
yield the syntactic capability to incrementally transform a proof to
the theorem it proves.  This could already be valuable in practice for
efficient proof checking, for example of large proofs produced by SAT
or SMT solvers (cf.~\cite{stump+12}).

It would be interesting to go further in automating proofs of type
preservation based on the rewriting approach.  While the Programming
Languages community has invested substantial effort in recent years on
computer-checked proofs of properties like type safety for programming
languages (initiated particularly by the POPLmark
Challenge~\cite{poplmark}), there is relatively little work on fully
automatic proofs of type preservation (an example
is~\cite{schurmann+98}).  The rewriting approach could contribute to
filling that gap, since the methods we used above for analyzing
interactions of abstract and concrete steps to prove type preservation
are similar to those used for proving confluence of combined
reduction.

Our longer term goal is to use this approach to design and analyze
type systems for symbolic simulation.  In program verification tools
like \textsc{Pex} and \textsc{KeY}, symbolic simulation is a central
component~\cite{KeyBook2007,Tillmann+2005}.  But these systems do not
seek to prove that their symbolic-simulation algorithms are correct.
Indeed, the authors of the \textsc{KeY} system argue against expending
the effort to do this~\cite{Beckert+06}.  The rewriting approach
promises to make it easier to relate symbolic simulation, viewed as an
abstract reduction relation, with the small-step operational
semantics.

\textbf{Acknowledgments.} We thank the anonymous LMCS reviewers for
their very detailed comments, and a number of technical suggestions
which have greatly improved this paper; and also participants of the
RTA 2011 conference for their helpful feedback and suggestions about
this work.


\end{document}